\documentclass[11pt,a4paper,psfig,times]{article}
\usepackage{latexsym}
\usepackage{epsfig}
\usepackage{times}
\usepackage{amssymb}
\usepackage{enumerate}
\usepackage{bm}
\usepackage{enumitem}
\usepackage{mathtools}
\usepackage{xfrac}
\usepackage{hyperref}

\usepackage[ruled,linesnumbered]{algorithm2e}

\usepackage{tikz}

\usepackage{amsthm, amsmath,amsfonts, amssymb}

\newcommand{\ignore}[1]{{}}

\def\p1{\phantom{1}}

\newtheorem{theorem}{Theorem}
\newtheorem{corollary}{Corollary}
\newtheorem{lemma}{Lemma}

\newtheorem{definition}{Definition}
\newtheorem{fact}{Fact}

\newcommand{\full}[1]{{}}

\begin{document}

\begin{titlepage}

\title{Tight Bounds for Online Matching in \\ Bounded-Degree Graphs with Vertex Capacities
%\footnotetext{This paper is typeset in 11pt front, with larger margins than required. We have five extra lines on page 12.}
}
\author{Susanne Albers\thanks{Department of Computer Science, Technical University of Munich. {\tt albers@in.tum.de}} \and Sebastian Schubert\thanks{Corresponding author. Department of Computer Science, Technical University of Munich. {\tt sebastian.schubert@tum.de}}}
\maketitle

\thispagestyle{empty}

%\maketitle

\begin{abstract}
We study the $b$-matching problem in bipartite graphs $G=(S,R,E)$. Each vertex $s\in S$ is a server with individual 
capacity $b_s$. The vertices $r\in R$ are requests that arrive online and must be assigned instantly to an eligible server. 
The goal is to maximize the size of the constructed matching. We assume that $G$ is a $(k,d)$-graph~\cite{NW}, where 
$k$ specifies a lower bound on the degree of each server and $d$ is an upper bound on the degree of each request. 
This setting models matching problems in timely applications.

We present tight upper and lower bounds on the performance of deterministic online algorithms. In particular, we develop 
a new online algorithm via a primal-dual analysis. The optimal competitive ratio tends to~1, for arbitrary $k\geq d$, 
as the server capacities increase. Hence, nearly optimal solutions can be computed online. Our results also hold for 
the vertex-weighted problem extension, and thus for AdWords and auction problems in which each bidder issues individual, 
equally valued bids.

Our bounds improve the previous best competitive ratios.
%for all $k$, $d$ and $b$. 
The asymptotic competitiveness of~1 
is a significant improvement over the previous factor of $1-1/e^{k/d}$, for 
the interesting range where $k/d\geq 1$ is small. Recall that $1-1/e\approx 0.63$. Matching problems that admit a competitive ratio arbitrarily close to~1 are rare. Prior results rely on randomization or
probabilistic input models. 
\end{abstract}

\end{titlepage}

\section{Introduction}
Maximum matching is a fundamental problem in computer science. In a seminal paper Karp, Vazirani and Vazirani~\cite{KVV} introduced
online matching in bipartite graphs $G=(S\cup R,E)$. The vertices of $S$ are known in advance, while the vertices of $R$ (requests)
arrive one by one and must be matched immediately to an eligible partner. The $b$-matching problem is a generalization where the 
vertices of $S$ (servers) have capacities and may be matched multiple times, see e.g.~\cite{KP}. Online bipartite matching
and capacitated extensions have received tremendous research interest over the past 30~years. In this paper we study the
$b$-matching problem in bounded-degree graphs, defined in~\cite{NW}. We assume that there is a lower bound on the degree of 
each server $s\in S$, meaning that there is a certain demand for each server. Furthermore we assume that there is an upper bound
on the degree of each $r\in R$, i.e.\ each request can only be assigned to a subset of the servers. This setting models 
matching problems in many timely applications, as we will describe below. 

More formally, we investigate the following problem. Again, let $G=(S\cup R,E)$ be a bipartite graph, where the vertices
of $S$ are servers and the vertices of $R$ are requests. The set $S$ is known in advance. Each server $s\in S$ has an
individual capacity $b_s\in \mathbb{N}$, indicating that the server can be matched with up to $b_s$ requests. The vertices
of $R$ arrive online, one by one.  Whenever a new request $r\in R$ arrives, its incident edges are revealed. The request 
has to be matched immediately and irrevocably to an eligible server, provided that there is one. The goal is to maximize 
the number of matching edges. We will also examine the \emph{vertex-weighted} problem extension, where additionally each 
server $s\in S$ has a weight $w_s$ and the value of every matching edge incident to $s$ is multiplied by $w_s$. Now
the goal is to maximize the total weight of the constructed matching. 

We assume that $G$ is a $(k,d)$-graph, defined by Naor and Wajc~\cite{NW}, where $k$ and $d$ are positive integers.
Each server $s\in S$ has a degree $d(s)\geq k\cdot b_s$. Each request $r\in R$ has a degree $d(r)\leq d$. Naor
and Wajc~\cite{NW} defined these graphs for the general AdWords problem. Note that the inequality $d(s)\geq k\cdot b_s$
expresses a degree bound in terms of the server capacity. This is essential. As we shall see, the performance of algorithms 
depends on the degrees $d(s)$ as a function of $b_s$. A degree bound independent of $b_s$ is vacuous for larger $b_s$.
Also, a company operating a high-capacity server expects the server to be attractive and a potential host for a large 
number of requests. Otherwise it might be beneficial to reduce the server capacity. 

The best results will be obtained if $k\geq d$. In this case the average demand for each server slot is high, compared to the
number of servers a request can be assigned to. This setting is also very relevant in applications. We will assume 
that $d\geq 2$. If $d=1$, any {\sc Greedy} algorithm constructs an optimal matching. We remark that $(k,d)$-graphs are loosely 
related to $d$-regular 
graphs in which each vertex has a degree of exactly $d$ and a capacity of~1. This graph class has been studied extensively in
computer science and discrete mathematics, see e.g.~\cite{CW,COS,CL,GKK,SCH}. 

The $b$-matching problem in $(k,d)$-graphs models many problems in modern applications, cf.~\cite{CDKL,GKKV,NW}. 
The following description also addresses the degree constraints. 

\smallskip

\emph{Video content delivery, web hosting, remote data storage}: Consider a collection of servers in a video content delivery
network, a web hosting provider, or a remote data storage service. A sequence of clients arrives, each with a request that
videos be streamed, web pages be hosted, or data be stored. Based on the servers' geographic distribution, average performance,
technology used or pricing policies, each request can only be hosted at a small subset of the servers or server 
locations. Each server has a large capacity and is well suited to service a huge number of requests in the arriving client
sequence. 

\smallskip

\emph{Job scheduling}: Consider a collection of compute servers, each with certain capabilities, located for instance in
a data center. Over a time horizon jobs arrive, requesting service. Based on computing demands, expected response time,
hardware and software requirements, each job can only be executed on a subset of the servers. During the given time horizon,
each server can process a large number of jobs and is a suitable platform to execute very many of the incoming jobs. 

\smallskip

\emph{AdWords and ad auctions}: Consider a search engine company or digital advertising platform. There is a set of 
advertisers, each with a daily budget, who wish to link their ads to users of the search engine/digital platform and
issue respective bids. The users arrive online and must be allocated instantly to the advertisers. Based on his search
keywords, browsing history and possible profile, each user is interesting to a small set of advertisers. 
Each advertiser has a decent budget and targets a large population of the users. Obviously, in this application the
advertisers correspond to the servers and the users are the incoming requests. The $b$-matching problem models the
basic setting where the bids of all advertisers are either~0 or 1. The vertex-weighted extension captures the scenario
where all the bids of an advertiser $s\in S$ have a value of~0 or $w_s$. These base cases are also studied in recent 
work by Vazirani~\cite{V}.

\smallskip

We analyze the performance of online algorithms using competitive analysis. Given an input graph $G$, let {\sc Alg}$(G)$ 
denote the size (or weight) of the matching constructed by an online algorithm {\sc Alg}. Let {\sc Opt}$(G)$ be the 
corresponding value of an optimal offline algorithm {\sc Opt}. Algorithm {\sc Alg} is $c$-competitive if 
$\textsc{Alg}(G) \geq c\cdot \textsc{Opt}(G)$ holds, for all $G$. In our analyses we will focus on 
bipartite $(k,d)$-graphs $G$. 

\vspace*{0.1cm}

{\bf Related work:}
As mentioned above, Karp et al.~\cite{KVV} introduced online matching in bipartite graphs, in which every vertex has a 
capacity of~1. The best competitive ratio of deterministic online algorithms is equal to $1/2$. Karp et al.\
proposed a randomized {\sc Ranking} algorithm that achieves an optimal competitive ratio of $1-1/e \approx 0.63$. 
Aggarwal et al.~\cite{AGKM} defined online vertex-weighted bipartite matching and devised a $(1-1/e)$-competitive algorithm.

Kalyanasundaram and Pruhs~\cite{KP} investigated the $b$-matching problem if all servers have equal capacity, i.e.\ $b_s=b$ for 
all $s\in S$. They presented a deterministic {\sc Balance} algorithm that matches a new request to an adjacent server whose 
current load is smallest. {\sc Balance} achieves an optimal competitive ratio of $1-1/(1+1/b)^b$. 
As $b$ grows, the latter expression tends from below to $1-1/e$.  Grove et al.~\cite{GKKV} and Chaudhuri et al.~\cite{CDKL} studied
$b$-matchings with a different objective. At any time an algorithm must maintain a matching between the requests that have
arrived so far and the servers. The goal is to minimize the total number of switches, reassigning requests to different 
servers.

The AdWords problem was formally defined by Mehta et al.~\cite{MSVV}. They presented a deterministic online algorithm that 
achieves a competitive ratio of $1-1/e$, under the \emph{small-bids assumption} where the bids are small compared to the 
advertisers' budgets. No randomized algorithm can obtain a better competitive factor. Buchbinder et al.~\cite{BJN} examined
a setting where the degree of each incoming user is upper bounded by $d$ and gave an algorithm with a competitive ratio of
nearly $1-(1-1/d)^d$. Azar et al.~\cite{ACR} showed that this ratio is best possible, also for randomized algorithms.
The expression $1-(1-1/d)^d$ is always greater than $1-1/e$ but approaches the latter value as $d$ increases. 

The class of $(k,d)$-graphs was defined by Naor and Wajc~\cite{NW}, who studied online bipartite matching and the AdWords
problem. They proposed an algorithm {\sc HighDegree} that matches a new request to an available neighbor of highest current 
degree. Naor and Wajc proved that {\sc HighDegree} and generalizations attain a competitive factor of $1-(1-1/d)^k$. This 
ratio holds for online bipartite matching and the vertex-weighted extension, where all vertices have a capacity of~1. 
Furthermore, it holds for AdWords with equal bids per bidder. For AdWords with arbitrary bids, the ratio is $(1-R_{\max})(1-(1-1/d)^k)$, 
where $R_{\max}$ is the maximum ratio between the bid of any bidder and its total budget. Naor and Wajc showed that no deterministic 
online algorithm for bipartite matching can achieve a competitive ratio greater than $1-(1-1/d)^k$ if $k\geq d$. For the general 
AdWords problem, they proved an upper bound of $(1-R_{\max})(1-(1-1/d)^{k/R_{\max}})$ if $k\geq d$. For increasing $k/d$, 
the expression $1-(1-1/d)^k$ tends to~1. For $k\approx d$ increasing, it approaches again $1-1/e$. 

Cohen and Wajc~\cite{CW} studied online bipartite matching in $d$-regular graphs and developed a randomized algorithm with 
a competitive ratio of $1- O(\sqrt{\log d/d})$, which tends to~1 as $d$ increases. 

Online bipartite matching and the AdWords problem have also been examined in stochastic input models. A random permutation of 
the vertices of $R$ may arrive. Alternatively, the vertices of $R$ are drawn i.i.d.\ from a known or unknown distribution. For 
online bipartite matching, the best online algorithms currently known achieve competitive ratios of 0.696 and 0.706~\cite{JL,MY}. 
The best possible performance ratios are upper bounded by 0.823~\cite{MOGS}, and hence bounded away from~1. 
For AdWords, $(1-\varepsilon)$-competitive algorithms are known, based on the small-bids assumption~\cite{DH,DSA}.

%\vspace*{0.1cm}

{\bf Our contributions:}
We present a comprehensive study of the $b$-matching problem in $(k,d)$-graphs. Specifically, we develop tight lower and upper bounds on 
the performance of deterministic online algorithms. The optimal competitive ratio tends to~1, for any choice of $k$ and $d$ with $k\geq d$, 
as the server capacities increase.

First, in Section~\ref{sec:alg} we investigate the setting that all servers have the same capacity, i.e. $b_s = b$ for all $s\in S$. 
We develop an optimal online algorithm {\sc WeightedAssignment} via a primal-dual analysis. The resulting strategy is simple. 
Associated with each server of current load $l$ and current degree $\delta$ is a value $V(l,\delta)$. An incoming request is assigned 
to an eligible server for which the increment $V(l,\delta+1)-V(l,\delta)$ is maximized. The values $V(l,\delta)$ can be calculated 
in a preprocessing step and retrieved by table lookup when the requests of $R$ are served. The best values $V(l,\delta)$, for variable 
$l$ and $\delta$, are determined using recurrence relations. Solving them is non-trivial because two parameters are involved.
% For $b=1$, our algorithm is equivalent to \textsc{HighDegree}.

We prove that {\sc WeightedAssignment} achieves a competitive ratio of $c^*$, where
\[c^* = 1 - \frac{1}{b} \left(\sum_{i=1}^b i \binom{kb}{b-i}{\frac{1}{(d-1)^{b-i}}}\right) \left(1-\frac{1}{d}\right)^{kb}.\] 
This is a slightly complex expression, but it is exact in all terms. In Section~\ref{sec:ub} we prove that no deterministic online 
algorithm can attain a competitive ratio greater than $c^*$, for any choices of k and d that fulfill $k \geq d$.

In Section~\ref{sec:ext} we consider two generalizations. We assume that each server $s\in S$ has an individual capacity $b_s$ and 
adapt {\sc WeightedAssignment}. As for the competitive factor, in $c^*$ the capacity $b$ has to be replaced by $b_{\min} := \min_{s\in S} b_s$. 
The resulting competitiveness is again optimal for $k\geq d$. Furthermore, we study the vertex-weighted problem extension and again adjust our 
algorithm. The competitive ratios are identical to those in the unweighted setting, for uniform and variable server capacities. 
Our results also hold for the AdWords problem where bidders issue individual, equally valued bids.

In Section~\ref{sec:comp} we analyze the optimal competitive ratio $c^*$. We prove that it tends to 1, for any $k\geq d$, 
as $b$ increases. Furthermore, we show that it is strictly increasing in $b$. % for any $k\geq d$. 
The analyses are involved and make non-trivial use of Gauss hypergeometric functions. 

A strength of our results is that the optimal competitiveness tends to~1, for increasing server capacities. Hence almost optimal 
solutions can be computed online. For the AdWords problem, high server capacities correspond to the small-bids assumption. Remarkably, 
in this setting near-optimal ad allocations can be computed based on structural properties of the input graph if bidders issue 
individual, equally valued bids. 
Recall that, without degree bounds, the competitive ratio for the $b$-matching problem tends from below to $1-1/e\approx 0.63$.
The competitiveness of $c^*$ improves upon the previous best ratio of $1-(1-1/d)^k$~\cite{NW}. The ratio $c^*$ is
equal to $1-(1-1/d)^k$ for $b=1$ and strictly increasing in $b$, for any $k\geq d$. Our asymptotic competitiveness of~1 is a significant improvement over $1-(1-1/d)^k \approx 1-1/e^{k/d}$, for the interesting range of small $k/d\geq 1$. For $k<d$, $1-(1-1/d)^k$ and $c^*$ can become small. The algorithms are still
$\frac{1}{2}$-competitive since they match requests whenever possible. We are aware of only two other online matching problems 
that admit competitive ratios arbitrarily close to~1. As mentioned above, a randomized algorithm is known for online matching in $d$-regular 
unit-capacity graphs~\cite{CW}. For the general AdWords problem, respective algorithms exist if the input $R$ is generated according to 
probability distributions~\cite{DH,DSA}.

\section{An optimal online algorithm}\label{sec:alg}
In this section we study the setting that all servers have a uniform capacity of $b$. We develop our algorithm \textsc{WeightedAssignment}.
While serving requests, the algorithm maintains a value $V(l_s,\delta_s)$, for each server $s$ with load $l_s$ and current degree $\delta_s$. 
At any point in time during the execution of the online algorithm, the load of a server $s$ denotes the amount of matched edges incident to $s$, while the current degree indicates the total number of edges incident to $s$.
In order to construct the function $V$ and for the purpose of analysis, we formulate \textsc{WeightedAssignment} as a primal-dual algorithm. 
The primal and dual linear programs of the $b$-matching problem are given at the top of the next page. 
The primal variables $m(s,r)$ indicate if an edge $\{s,r\}\in E$ belongs to the matching. We have dual variables $x(s)$ and $y(r)$.

In the pseudocode of \textsc{WeightedAssignment}, also shown on this page, Line~7 is the actual matching step. A new request
$r$ is assigned to a neighboring server $s$ for which the difference ${V(l_s,\delta_s+1)}-V(l_s,\delta_s)$ is maximized. 
$N(r)$ is the set of adjacent servers with remaining capacity. 
All other instructions essentially update primal and dual variables so that a primal and a dual solution 
are constructed in parallel. 

Observe that no dual variable $y(r)$ of any request $r$ is ever increased by \textsc{WeightedAssignment}. The dual variable $x(s)$ of a server $s$ can be increased in Lines 9 and 11. It is increased if $s$ is matched to a neighboring request $r$ and, importantly, $x(s)$ is also increased if this $r$ is assigned
to a different server. 
\begin{align*}
\textbf{P: }\text{max} \ &\sum_{\{s,r\}\in E} m(s,r) 
& \textbf{D: }\text{min} \ &\sum_{s\in S} b \cdot x(s) + \sum_{r\in R} y(r) \\
\text{s.t.} \ &\sum_{r:\{s,r\}\in E} m(s,r) \leq  b, \ (\forall s \in S)
&\text{s.t.} \ & x(s)+y(r) \geq 1, \ (\forall \{s,r\}\in E) \\
&\sum_{s:\{s,r\}\in E} m(s,r) \leq 1, \ (\forall r\in R)
& &x(s),\,y(r)\geq 0, \ (\forall s\in S, \forall r\in R) \\
& m(s,r) \geq 0, \ (\forall \{s,r\}\in E)
\end{align*}
%\textsc{WeightedAssignment} is a primal-dual algorithm, meaning that it constructs a primal and dual solution in parallel. A function $V$ is used for the definition. The %algorithm ensures that the value of the dual variable $x(s)$ of a server $s\in S$ is equal to $V(l_s,\delta_s)$ when $s$ has load $l_s$ and degree $\delta_s$ during the %execution. In the analysis, we will see how the function $V$ has to be constructed, such that \textsc{WeightedAssignment} achieves a competitive ratio of 
%\[
%    c^* = 1-\frac{1}{b}\left(\sum_{i=1}^{b} i \binom{kb}{b-i}\frac{1}{(d-1)^{b-i}}\right)\left(1-\frac{1}{d} \right)^{kb} \,.
%\]
%We formally prove that $c^*$ is monotonically increasing in $b$ and converges to $1$ if $b\rightarrow \infty$ in Section~\ref{sec:comp}.
%In Algorithm \ref{alg:name}, $l_s$ and $\delta_s$ denote the current load and degree of server $s$ throughout the execution of the algorithm, respectively. 
\begin{algorithm}[H]
\caption{\textsc{WeightedAssignment}}
  \label{alg:WA}
  \SetAlgoLined
    Initialize $x(s)=0$, $y(r)=0$ and $m(s,r)=0$, $\forall s\in S$ and $\forall r\in R$\; 
	\While{a new request $r\in R$ arrives} {
%	  Let $N(r)$ denote the set of neighbors $s$ of $r$ with $x(s)<1$\;
      Let $N(r)$ denote the set of neighbors $s$ of $r$ with remaining capacity\;
	 \eIf{$N(r)=\emptyset$}{
	  Do not match $r$\;}{
	  Match $r$ to $\arg \max\left\{V(l_s,\delta_s+1)-V(l_s,\delta_s): s\in N(r)\right\}$\;
	  Update $m(s,r)\leftarrow 1$\;
	  Set $x(s)\leftarrow V(l_s+1,\delta_{s}+1)$\;
	  \ForAll{$s'\neq s \in N(r)$} {
	    Set $x(s') \leftarrow V(l_{s'},\delta_{s'}+1)$\;
	  }
	 }
	}
\end{algorithm}
\smallskip

In the analysis, we will see how the function $V$ has to be defined so that \textsc{WeightedAssignment} achieves the desired competitive ratio $c^*$. 
Note that we always construct a feasible dual solution if $x(s)=1$ holds, for all servers $s\in S$, by the end of the algorithm. Here lies a crucial idea of the algorithm and the construction of $V$. We demand $V(b,\delta)=1$, for all $\delta \geq b$, and $V(l, \delta)= 1$ if $\delta\geq kb$, for all $0\leq l \leq b$. Also, $V(0,0)=0$. These constraints have two important implications.
\begin{enumerate}
    \item The dual variable $x(s)$ of a server $s$ with load $l_s$ and degree $\delta_s$ is always equal to $V(l_s,\delta_s)$:
    Consider an incoming request $r$ that is a neighbor of $s$. While $l_s<b$, it holds $N(r)\neq\emptyset$ and $r$ is matched to 
    some server. Lines~9 and 11 correctly update $x(s)$ with respect to the new load and degree values. If $l_s=b$, then 
    inductively by the first constraint $x(s)=1$ and no update is necessary.
    %It is easy to see that if an incoming request $r$ is matched by \textsc{WeightedAssignment}, then all neighbors of $r$ with remaining capacity are updated correctly %according to $V$ with respect to their load and degree. However, if $r$ remains unmatched, no dual variable of any neighbor is updated. But, $r$ remaining unmatched %implies that all neighbors of $r$ have to be full. Hence, we have to make sure that no updates to the dual variables of full servers are necessary, which is captured %by the first constraint. 
    \item The constructed dual solution is feasible: Implication 1 and the second constraint ensure that ${x(s)=1}$ holds for all $s\in S$ by the end of the algorithm, since every server $s$ has a degree of at least $kb$.
\end{enumerate}

Let $P$ and $D$ denote the value of the primal and dual solution constructed by the algorithm, respectively. We denote a change in the value of the primal and dual solution by $\Delta P$ and $\Delta D$, respectively. It holds that the size of the matching constructed by \textsc{WeightedAssignment} is exactly $|\textsc{Alg}|=P$. Moreover, by weak duality, we get that the size of the optimum matching is $|\textsc{Opt}|\leq D$. Hence, if we were able to bound $\Delta D \leq \Delta P/c$ at every step, we would obtain a competitive ratio of $c$. 
\[
 \frac{|\textsc{Alg}|}{|\textsc{Opt}|} \geq \frac{P}{D} \geq \frac{P}{\frac{1}{c}\cdot P} = c \,.
\]
Recall that the value of the dual solution is only increased if a request $r$ is matched to a server $s$. Then, the value of the primal solution is increased by 1, while the value of the dual solution is increased by 
\[
    \Delta D = b\cdot \Big( V(l_s+1,\delta_{s}+1) - V(l_s,\delta_{s})+ \sum_{s'\in N(r)\setminus\{s\}} V(l_{s'},\delta_{s'}+1) - V(l_{s'},\delta_{s'})\Big) \,.
\]
The algorithm chooses $s$ such that $V(l_s,\delta_s+1)-V(l_s,\delta_s) \geq V(l_{s'},\delta_{s'}+1)-V(l_{s'},\delta_{s'})$ holds for all $s'\in N(r)$. Furthermore, $|N(r)|\leq d$ implies that we can bound this increase by 
\[
    \Delta D \leq b\cdot \Big( V(l_s+1,\delta_{s}+1) - V(l_s,\delta_{s}) + (d-1)\cdot \big(V(l_{s},\delta_{s}+1) - V(l_{s},\delta_{s})\big)\Big)\,.
\]
This means, that we need to determine the biggest possible constant $c^*\in (0,1]$ such that
\begin{equation}
\label{equ:temp}
    b\cdot \Big( V(l+1,\delta+1) - V(l,\delta) + (d-1)\cdot \big(V(l,\delta+1) - V(l,\delta)\big)\Big) \leq \frac{1}{c^*} 
\end{equation}
holds for all $0\leq l<b$ and all $\delta \geq l$, while still satisfying our constraints that $V(b,\cdot)= 1$ and $V(\cdot,\delta')= 1$ for $\delta'\geq kb$. For this, we define
\begin{equation*}
    p(l,\delta) := V(l+1,\delta+1) - V(l,\delta) \quad \text{and} \quad q(l,\delta) := V(l,\delta+1) - V(l,\delta) \,.
\end{equation*}
In other words, the dual variable $x(s)$ of a server $s$ with load $l$ and current degree $\delta$ is increased by $p(l,\delta)$, when a request is assigned to $s$, and increased by $q(l,\delta)$, when a neighboring request is assigned to a different server. Our constraints immediately give $p(l,\delta)=q(l,\delta)=0$, if $l=b$ or $\delta \geq kb$. Hence, we will focus on the case $0\leq l< b$ and $l\leq \delta <kb$ in the following. Rewriting and rearranging inequality~(\ref{equ:temp}) in terms of $p$ and $q$ yields
\begin{equation*}
\label{equ:main}
    q(l,\delta)  \leq \frac{1}{d-1}\left(\frac{1}{b\cdot c} - p(l,\delta) \right) \, .
\end{equation*}
We treat the values $p(i,i)$, for $0\leq i <b-1$, as the variables of our optimization, since every other $p$ and $q$ value can then be computed based on these choices. To get comfortable with the recursions and ideas in the latter part of this section, we do the following warm-up, where we consider $V(b-1,\delta)$. It holds 
\[
    V(b-1,\delta)= \sum_{i=0}^{b-2} p(i,i) + \sum_{j=b-1}^{\delta-1} q(b-1,j) \,.
\] 

Our first constraint $V(b,\delta)=1$, for all $\delta\geq b$, implies that $p(b-1,\delta)=1-V(b-1,\delta)$. We do not want to waste any potential increases in our dual variables, since we want to maximize $c$. Thus, we will choose the maximum possible value for $q(b-1,\delta)$, which is 
\[
q(b-1,\delta)  = \frac{1}{d-1}\left(\frac{1}{b\cdot c} - p(b-1,\delta) \right) = \frac{1}{d-1}\left(\frac{1}{b\cdot c} - 1+V(b-1,\delta) \right) \,.
\]
It follows that 
\begin{equation}
    \label{equ:recurrence1}
    V(b-1,\delta+1) = V(b-1,\delta) + q(b-1,\delta) = \frac{d}{d-1} V(b-1,\delta) + \frac{1}{d-1}\left(\frac{1}{b\cdot c} -1\right) \,,
\end{equation}
for all $b-1 \leq \delta < kb$ and with $V(b-1,b-1)=\sum_{i=0}^{b-2} p(i,i)$. To ease notation in the future, we further define $P_i := \sum_{j=0}^{i-1} p(j,j)$, for all $0\leq i \leq b-1$,
%\begin{equation*}
%    P_i := \sum_{j=0}^{i-1} p(j,j) \,,
%\end{equation*}
so that we get $P_i=V(i,i)$. Note that $P_0=0$. Solving the recurrence relation (\ref{equ:recurrence1}) yields 
\begin{align*}
    V(b-1,\delta) &= \left(\frac{d}{d-1} \right)^{\delta-(b-1)} P_{b-1} + \frac{1}{d-1}\left(\frac{1}{b\cdot c} -1\right)\cdot \sum_{i=0}^{\delta-(b-1)-1} \left(\frac{d}{d-1} \right)^i \\
    &= \left(\frac{d}{d-1} \right)^{\delta-(b-1)} P_{b-1} + \frac{1}{d-1}\left(\frac{1}{b\cdot c} -1\right)\cdot  \frac{\left(\frac{d}{d-1} \right)^{\delta-(b-1)}-1}{\left(\frac{d}{d-1} \right)-1} \\
    &= \left(\frac{d}{d-1} \right)^{\delta-(b-1)} \left( P_{b-1} + \frac{1}{b\cdot c}-1\right) +1 - \frac{1}{b\cdot c} \,.
\end{align*}
The following lemma generalizes the computation above to all other load levels. 

\begin{lemma}
\label{lem:V}
For all $l$, $0\leq l \leq b$, and for all $\delta$, $l\leq \delta\leq kb$, it holds that
\begin{equation}
    \label{equ:Vsolved}
    V(l,\delta) = \sum_{i=l}^{b-1}(-1)^{i-l} \frac{1}{(d-1)^{i-l}}\binom{\delta-l}{i-l}\left(\frac{d}{d-1}\right)^{\delta-i}\left(P_i+\frac{b-i}{b \cdot c}-1 \right) +1-\frac{b-l}{b\cdot c} \,.
\end{equation}
\end{lemma}

\begin{proof}
By induction over $l$, starting with $l=b$ and going down to $l=0$. The induction base $l=b$ is true, because we have $V(b,\delta)=1$. Thus, we focus on the induction step $l+1\leadsto l$. Similar arguments as before yield for $l\leq \delta < kb$
\begin{equation*}
    q(l,\delta) = \frac{1}{d-1}\left(\frac{1}{b\cdot c} - p(l,\delta) \right) = \frac{1}{d-1}\left(\frac{1}{b\cdot c} - V(l+1,\delta+1) + V(l,\delta) \right) \,.
\end{equation*}
We can now define the recurrence relation for $V(l,\delta)$
\begin{align*}
    V(l,\delta+1) = V(l,\delta) + q(l,\delta) = \frac{d}{d-1} V(l,\delta) + \frac{1}{d-1}\left(\frac{1}{b\cdot c} - V(l+1,\delta+1) \right) \,,
\end{align*}
with $V(l,l)=P_l$. Solving this recurrence yields 
\begin{equation}
    \label{equ:Vtemp}
    \begin{aligned}
    V(l,\delta) = &\left(\frac{d}{d-1}\right)^{\delta-l} P_l + \frac{1}{d-1} \frac{1}{b\cdot c} \sum_{i=0}^{\delta-l-1} \left(\frac{d}{d-1}\right)^i 
    %\\ &
    - \frac{1}{d-1} \sum_{i=0}^{\delta-l-1} \left(\frac{d}{d-1}\right)^i V(l+1,\delta-i) \,.
    \end{aligned}
\end{equation}
In the next step, we will need the following fact~\cite{GKP}.
\begin{fact}
\label{fac:binom}
For $n,k\in \mathbb{N}_0$, it holds that 
\[
\sum_{i=0}^n \binom{i}{k} = \binom{n+1}{k+1} \,.
\]
\end{fact}

Next, we apply the induction hypothesis to determine $V(l,\delta)$. For clarity, we focus on the last sum of equality~(\ref{equ:Vtemp}) first 
\begin{align*}
    &\sum_{i=0}^{\delta-l-1}  \left(\frac{d}{d-1}\right)^i  V(l+1,\delta-i) \overset{\text{IH}}{=}\sum_{i=0}^{\delta-l-1}  \left(\frac{d}{d-1}\right)^i  \left[ 
    \sum_{j=l+1}^{{b}-1}(-1)^{j-(l+1)} \frac{1}{(d-1)^{j-(l+1)}}
    \right.\\ &\left. \quad \cdot \,
    \binom{\delta-i-(l+1)}{j-(l+1)} \left(\frac{d}{d-1}\right)^{\delta-i-j}\left(P_j+\frac{{b}-j}{b\cdot c}-1\right)
    %\right.\\ &\left. \phantom{\sum_{j=l+1}^{{b}-1}} 
    + 1-\frac{{b}-(l+1)}{b\cdot c}\right]  \\
    &=\sum_{j=l+1}^{{b}-1}  (-1)^{j-(l+1)}  \frac{1}{(d-1)^{j-(l+1)}}\left(\frac{d}{d-1}\right)^{\delta-j}\left(P_j+\frac{{b}-j}{b\cdot c}-1\right)\sum_{i=0}^{\delta-l-1} \binom{\delta-i-(l+1)}{j-(l+1)} \\
        & \quad + \left(1-\frac{{b}-(l+1)}{b\cdot c}\right) \sum_{i=0}^{\delta-l-1} \left(\frac{d}{d-1}\right)^i  \\
    &= \sum_{j=l+1}^{{b}-1} (-1)^{j-(l+1)}  \frac{1}{(d-1)^{j-(l+1)}}\left(\frac{d}{d-1}\right)^{\delta-j}\left(P_j+\frac{{b}-j}{b\cdot c}-1\right)\sum_{i=0}^{\delta-l-1} \binom{i}{j-(l+1)} \\
        & \quad + \left(1-\frac{b-(l+1)}{b\cdot c}\right) (d-1)\left( \left(\frac{d}{d-1} \right)^{\delta-l}-1 \right)  \\
    &= \sum_{j=l+1}^{b-1} (-1)^{j-(l+1)}  \frac{1}{(d-1)^{j-(l+1)}}\left(\frac{d}{d-1}\right)^{\delta-j}\left(P_j+\frac{{b}-j}{b\cdot c}-1\right)\binom{\delta -l}{j-l}  \\
        & \quad + \left(1-\frac{b-(l+1)}{b\cdot c}\right) (d-1)\left( \left(\frac{d}{d-1} \right)^{\delta-l}-1 \right)  \,,
\end{align*}
where we used Fact \ref{fac:binom} in the last step. Now, we can finish the induction step by plugging this into (\ref{equ:Vtemp})
\begin{align*}
   V(l,\delta) &= \left(\frac{d}{d-1}\right)^{\delta-l} P_l +  \frac{1}{b\cdot c} \left( \left(\frac{d}{d-1} \right)^{\delta-l}-1 \right)   - \left( 1-\frac{b-(l+1)}{b\cdot c}\right) \left( \left(\frac{d}{d-1} \right)^{\delta-l}-1 \right)\\
    & \quad 
    + \sum_{j=l+1}^{b-1} (-1)^{j-l}  \frac{1}{(d-1)^{j-l}}\left(\frac{d}{d-1}\right)^{\delta-j}\left(P_j+\frac{{b}-j}{b\cdot c}-1 \right) \binom{\delta -l}{j-l}  \\
    %& \quad - \left( 1-\frac{b-(l+1)}{b\cdot c}\right) \left( \left(\frac{d}{d-1} \right)^{\delta-l}-1 \right) \\
    &= \left(\frac{d}{d-1}\right)^{\delta-l} \left(P_l+\frac{1}{b\cdot c}+\frac{b-(l+1)}{b\cdot c}-1 \right)+1-\frac{1}{b\cdot c}-\frac{b-(l+1)}{b\cdot c}\\ 
     & \quad + \sum_{j=l+1}^{b-1} (-1)^{j-l}  \frac{1}{(d-1)^{j-l}}\left(\frac{d}{d-1}\right)^{\delta-j}\left(P_j+\frac{{b}-j}{b\cdot c}-1 \right) \binom{\delta -l}{j-l}  \\
   &= \sum_{j=l}^{b-1} (-1)^{j-l}  \frac{1}{(d-1)^{j-l}}\left(\frac{d}{d-1}\right)^{\delta-j}\left(P_j+\frac{{b}-j}{b\cdot c}-1\right)\binom{\delta -l}{j-l} + 1 - \frac{b-l}{b\cdot c}\,. \qedhere
\end{align*}
\end{proof}

So far, we have only leveraged our constraint $V(b,\cdot)=1$. With our description of $V(l,\delta)$, for all $0\leq l \leq b$ and $l\leq \delta \leq kb$, we can also leverage $V(\cdot, kb)=1$ to determine $P_i$, for all $0\leq i\leq b-1$. For this, we will need the following technical lemma. 
%The proof is given in Appendix~\ref{app:secalg}.

\begin{lemma}
\label{lem:binom}
For $k,n,m \in \mathbb{N}$, with $m\geq n \geq k$, it holds that 
\begin{align*}
    \sum_{i=1}^{n-k} (-1)^i \binom{m}{i}  \binom{m-i}{n-k-i} = - \binom{m}{n-k} \, .
\end{align*}
\end{lemma}

\begin{proof}
Observe that 
\begin{align*}
    \binom{m}{i}  \binom{m-i}{n-k-i} &= \frac{m!}{(m-i)!i!} \cdot \frac{(m-i)!}{(m-(n-k))!(n-k-i)!} \\
        &=  \frac{(n-k)!}{(n-k-i)!i!} \cdot \frac{m!}{(m-(n-k))!(n-k)!} = \binom{n-k}{i} \binom{m}{n-k}\, .
\end{align*}
Next, notice that the binomial expansion implies that
\begin{equation*}
    0 = (1-1)^{n-k} = \sum_{i=0}^{n-k}\binom{n-k}{i}(-1)^i \,.
\end{equation*}
Hence, we can conclude
\begin{align*}
    \sum_{i=1}^{n-k} (-1)^i \binom{m}{i}  \binom{m-i}{n-k-i} &= \binom{m}{n-k}  \sum_{i=1}^{n-k} (-1)^i \binom{n-k}{i} \\ &=  \binom{m}{n-k} \left(\sum_{i=0}^{n-k} (-1)^i \binom{n-k}{i} -1\right) \\ &= \binom{m}{n-k} (0-1) = -\binom{m}{n-k} \, . \qedhere
\end{align*}
\end{proof}

\begin{lemma}
\label{lem:P}
For all $l$, $0\leq l \leq b-1$, it holds that
\begin{equation}
    \label{equ:Pbound}
   \left(\frac{d}{d-1} \right)^{kb-l} \left(P_l+\frac{b-l}{b\cdot c}-1\right) = \frac{1}{b\cdot c}\left(\sum_{i=1}^{b-l} i \binom{kb-l}{b-l-i} \frac{1}{(d-1)^{b-l-i}} \right)\,.
\end{equation}
\end{lemma}

\begin{proof}
By induction over $l$ from $l=b-1$ down to $l=0$. We start with the induction base $l=b-1$. Our second constrain yields $V(b-1,kb)= 1$. It then follows from Lemma~\ref{lem:V} that
\begin{align*}
    V(b-1,kb) = \left(\frac{d}{d-1}\right)^{kb-(b-1)}\left(P_{b-1}+\frac{1}{b\cdot c}-1 \right) +1-\frac{1}{b\cdot c} \overset{\text{def.}}{=} 1 \, , 
\end{align*}
which immediately finishes the induction base, since the right-hand side of (\ref{equ:Pbound}) is simply $1/(b\cdot c)$.

We can now move on to the induction step $\forall i>l \leadsto l$. We use $V(l,kb)= 1$ and rearrange with the help of Lemma~\ref{lem:V}
\begin{equation}
    \label{equ:Ptemp}
    \begin{aligned}
    &\left(\frac{d}{d-1} \right)^{kb -l}\left(P_l+\frac{b-l}{b\cdot c}-1\right) = \frac{b-l}{b\cdot c} 
    \\ &\quad - \sum_{i=l+1}^{b-1}(-1)^{i-l} \frac{1}{(d-1)^{i-l}}\binom{kb-l}{i-l}\left(\frac{d}{d-1}\right)^{kb-i}\left(P_i+\frac{b-i}{b\cdot c}-1\right)\,.
    \end{aligned}
\end{equation}
It is now possible to apply the induction hypothesis. For clarity, we focus on the second line of (\ref{equ:Ptemp})
\begin{align*}
    & \sum_{i=l+1}^{b-1}(-1)^{i-l} \frac{1}{(d-1)^{i-l}}\binom{kb-l}{i-l} \frac{1}{b\cdot c}\left(\sum_{j=1}^{b-i} j \binom{kb-i}{b-i-j} \frac{1}{(d-1)^{b-i-j}} \right) \\
    & =\frac{1}{b\cdot c}\left( \sum_{a=1}^{b-l-1} (-1)^a \sum_{j=1}^{b-(l+a)} j \frac{1}{(d-1)^{b-l-j}} \binom{kb-l}{a}\binom{kb-(l+a)}{b-(l+a)-j} \right)\, ,
\end{align*}
where we substituted $a:=i-l$. Next, we carefully swap these nested sums. For this, observe that, for a fixed value $j$, we have exactly $b-l-j$ addends, more precisely, one addend for each $1\leq a \leq b-l-j$ 
\begin{align*}
    & \frac{1}{b\cdot c}\left( \sum_{a=1}^{b-l-1} (-1)^a \sum_{j=1}^{b-
(l+a)} j \frac{1}{(d-1)^{b-l-j}} \binom{kb-l}{a}\binom{kb-l-a}{b-l-a-j} \right) \\ 
    & =  \frac{1}{b\cdot c} \left(\sum_{j=1}^{b-l-1} j \frac{1}{(d-1)^{b-l-j}} \sum_{a=1}^{b-l-j} (-1)^a \binom{kb-l}{a}\binom{kb-l-a}{b-l-a-j} \right) \\ 
    & = -\frac{1}{b\cdot c}\left(\sum_{j=1}^{b-l-1} j \frac{1}{(d-1)^{b-l-j}}\binom{kb-l}{b-l-j}\right) \,,
\end{align*}
where we applied Lemma~\ref{lem:binom} with $k=j$, $n=b-l$ and $m=kb-l$ in the last step. 
Plugging this back into (\ref{equ:Ptemp}) gives 
\begin{align*}
    \left(\frac{d}{d-1} \right)^{kb -l}\left(P_l+\frac{b-l}{b\cdot c}-1\right)&=\frac{b-l}{b\cdot c} + \frac{1}{b\cdot c} \left(\sum_{j=1}^{b-l-1} j \frac{1}{(d-1)^{b-l-j}}\binom{kb-l}{b-l-j}\right) \\ 
    &= \frac{1}{b\cdot c}\left(\sum_{j=1}^{b-l} j \frac{1}{(d-1)^{b-l-j}}\binom{kb-l}{b-l-j}\right)
    \,. \qedhere
\end{align*}
\end{proof}

With the help of Lemma \ref{lem:P}, we can finally determine the resulting constant $c^*$. For $l=0$, we have $P_0=0$, and thus 
\[    
\left(\frac{d}{d-1} \right)^{kb}\left(\frac{1}{c^*}-1\right) = \frac{1}{b\cdot c^*}\left(\sum_{i=1}^{b} i \binom{kb}{b-i} \frac{1}{(d-1)^{b-i}} \right)\,,
\]
where solving for $c^*$ yields 
\[
    c^* =1-\frac{1}{b}\left(\sum_{i=1}^{b} i \binom{kb}{b-i}\frac{1}{(d-1)^{b-i}}\right)\left(1-\frac{1}{d} \right)^{kb}\,.
\]

%To summarize, we have shown how to construct a function $V$, where $V(l,\delta)$ determines the value of the dual variable of a server when it has load $l$ and degree $\delta$ during the execution of \textsc{WeightedAssignment}. Moreover, we have shown that, whenever \textsc{WeightedAssignment} matches a request $r$ to $s$, we can bound the increase in the value of the dual solution by
%\begin{equation*}
%    \Delta D \leq b\cdot \Big( V(l+1,\delta+1) - V(l,\delta) + (d-1)\cdot \big(V(l,\delta+1) - V(l,\delta)\big)\Big) \leq \frac{1}{c^*} \, ,
%\end{equation*}
%for all $0\leq l<b$, and for all $0\leq \delta<kb$. Furthermore, we have $V(b,\cdot)=1$ and $V(\cdot, \delta')=1$, for all $\delta' \geq kb$, meaning that we indeed construct feasible primal and dual solutions. This results in the following theorem. 

\begin{theorem}
\textsc{WeightedAssignment} achieves a competitive ratio of $c^*$ for the
online $b$-matching problem with uniform server capacities on $(k,d)$-graphs. 
\end{theorem}

\begin{figure}[bth!]
    \centering
\begin{tikzpicture}[scale=1]
\foreach \x in {0, 1,..., 8} {
    \node at (3*\x/2,-1) {\x};
}
\foreach \y in {0, 1,..., 4} {
    \node at (-1,3/2*\y) {\y};
}
\node at (-2,3) {$l$};
\node at (3/2*4,-2) {$\delta$};

\node (00) at (0,0) {0};
\node (10) at (3/2*1,0) {16};
\node (20) at (3/2*2,0) {37};
\node (30) at (3/2*3,0) {63};
\node (40) at (3/2*4,0) {93};
\node (50) at (3/2*5,0) {125};
\node (60) at (3/2*6,0) {157};
\node (70) at (3/2*7,0) {189};
\node (80) at (3/2*8,0) {221};

\node (11) at (3/2*1,3/2*1) {48};
\node (21) at (3/2*2,3/2*1) {59};
\node (31) at (3/2*3,3/2*1) {75};
\node (41) at (3/2*4,3/2*1) {97};
\node (51) at (3/2*5,3/2*1) {125};
\node (61) at (3/2*6,3/2*1) {157};
\node (71) at (3/2*7,3/2*1) {189};
\node (81) at (3/2*8,3/2*1) {221};

\node (22) at (3/2*2,3/2*2) {101};
\node (32) at (3/2*3,3/2*2) {107};
\node (42) at (3/2*4,3/2*2) {117};
\node (52) at (3/2*5,3/2*2) {133};
\node (62) at (3/2*6,3/2*2) {157};
\node (72) at (3/2*7,3/2*2) {189};
\node (82) at (3/2*8,3/2*2) {221};
    
\node (33) at (3/2*3,3/2*3) {159};
\node (43) at (3/2*4,3/2*3) {161};
\node (53) at (3/2*5,3/2*3) {165};
\node (63) at (3/2*6,3/2*3) {173};
\node (73) at (3/2*7,3/2*3) {189};
\node (83) at (3/2*8,3/2*3) {221};

\node (44) at (3/2*4,3/2*4) {221};
\node (54) at (3/2*5,3/2*4) {221};
\node (64) at (3/2*6,3/2*4) {221};
\node (74) at (3/2*7,3/2*4) {221};
\node (84) at (3/2*8,3/2*4) {221};

\draw[->] (00) to node[above] {\footnotesize 16} (10);
\draw[->] (10) to node[above] {\footnotesize 21} (20);
\draw[->] (20) to node[above] {\footnotesize 26} (30);
\draw[->] (30) to node[above] {\footnotesize 30} (40);
\draw[->] (40) to node[above] {\footnotesize 32} (50);
\draw[->] (50) to node[above] {\footnotesize 32} (60);
\draw[->] (60) to node[above] {\footnotesize 32} (70);
\draw[->] (70) to node[above] {\footnotesize 32} (80);

\draw[->] (00) to node[above left] {\footnotesize 48} (11);
\draw[->] (10) to node[above left] {\footnotesize 43} (21);
\draw[->] (20) to node[above left] {\footnotesize 38} (31);
\draw[->] (30) to node[above left] {\footnotesize 34} (41);
\draw[->] (40) to node[above left] {\footnotesize 32} (51);
\draw[->] (50) to node[above left] {\footnotesize 32} (61);
\draw[->] (60) to node[above left] {\footnotesize 32} (71);
\draw[->] (70) to node[above left] {\footnotesize 32} (81);

\draw[->] (11) to node[above] {\footnotesize 11} (21);
\draw[->] (21) to node[above] {\footnotesize 16} (31);
\draw[->] (31) to node[above] {\footnotesize 22} (41);
\draw[->] (41) to node[above] {\footnotesize 28} (51);
\draw[->] (51) to node[above] {\footnotesize 32} (61);
\draw[->] (61) to node[above] {\footnotesize 32} (71);
\draw[->] (71) to node[above] {\footnotesize 32} (81);

\draw[->] (11) to node[above left] {\footnotesize 53} (22);
\draw[->] (21) to node[above left] {\footnotesize 48} (32);
\draw[->] (31) to node[above left] {\footnotesize 42} (42);
\draw[->] (41) to node[above left] {\footnotesize 36} (52);
\draw[->] (51) to node[above left] {\footnotesize 32} (62);
\draw[->] (61) to node[above left] {\footnotesize 32} (72);
\draw[->] (71) to node[above left] {\footnotesize 32} (82);

\draw[->] (22) to node[above] {\footnotesize 6} (32);
\draw[->] (32) to node[above] {\footnotesize 10} (42);
\draw[->] (42) to node[above] {\footnotesize 16} (52);
\draw[->] (52) to node[above] {\footnotesize 24} (62);
\draw[->] (62) to node[above] {\footnotesize 32} (72);
\draw[->] (72) to node[above] {\footnotesize 32} (82);

\draw[->] (22) to node[above left] {\footnotesize 58} (33);
\draw[->] (32) to node[above left] {\footnotesize 54} (43);
\draw[->] (42) to node[above left] {\footnotesize 48} (53);
\draw[->] (52) to node[above left] {\footnotesize 40} (63);
\draw[->] (62) to node[above left] {\footnotesize 32} (73);
\draw[->] (72) to node[above left] {\footnotesize 32} (83);

\draw[->] (33) to node[above] {\footnotesize 2} (43);
\draw[->] (43) to node[above] {\footnotesize 4} (53);
\draw[->] (53) to node[above] {\footnotesize 8} (63);
\draw[->] (63) to node[above] {\footnotesize 16} (73);
\draw[->] (73) to node[above] {\footnotesize 32} (83);

\draw[->] (33) to node[above left] {\footnotesize 62} (44);
\draw[->] (43) to node[above left] {\footnotesize 60} (54);
\draw[->] (53) to node[above left] {\footnotesize 56} (64);
\draw[->] (63) to node[above left] {\footnotesize 48} (74);
\draw[->] (73) to node[above left] {\footnotesize 32} (84);

\end{tikzpicture}
    \caption{The function $V$ for $k=d=2$ and $b=4$. All values are multiplied by $221$.}
    \label{fig:Vexample}
\end{figure}
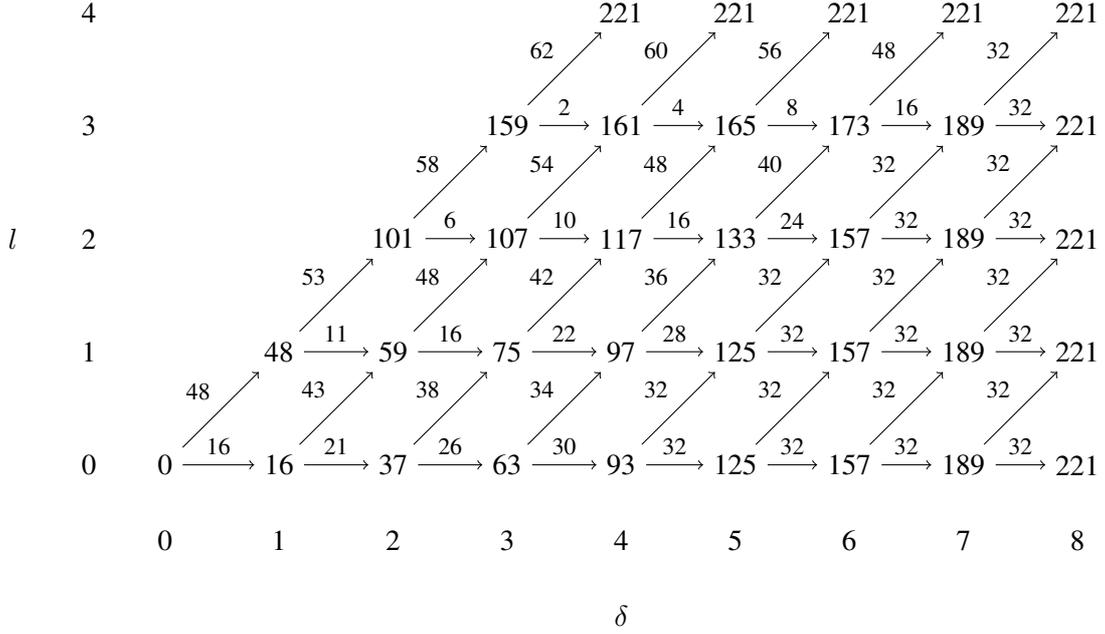

Lemma~\ref{lem:V}, together with Lemma~\ref{lem:P} and $c^*$, specifies the function $V$. Its values can be calculated in a preprocessing step and accessed
by table lookup when {\sc WeightedAssignment} serves requests. The parameters $k$ and $d$ must be known. In an application, they can be learned over
time. Alternatively, one can work with conservative estimates. 
Figure~\ref{fig:Vexample} shows the function $V$ for a small example with $k=d=2$ and $b=4$. In this case, we have 
$c^*=221/256$ and $1/(bc^*) = 64/221$. The arrows depict the possible increases in $V$ for every $l$, $0\leq l<b$, and $\delta$, $l\leq \delta<kb$, 
i.e.\ horizontal arrows denote the $q(l,\delta)$ values and diagonal arrows the $p(l,\delta)$ values. All actual values are multiplied by $221$ 
in order to eliminate fractions. Observe that $p(l,\delta)+(d-1)q(l,\delta)=1/(bc^*)$ holds for all $l$ and $\delta$. 
This is what allows us to bound the total increase in the dual solution by $1/c^*$, if we always pick the neighboring server $s$ that maximizes $q(l_s, \delta_s)$. 
%Furthermore, we have $V(b,\cdot)=V(\cdot,kb)=1$, which satisfies our constraints. 
For the matching decisions, {\sc WeightedAssignment} only uses the horizontal arrows. 
Notice that {\sc WeightedAssignment} is different from a \textsc{Balance} algorithm that breaks ties by \textsc{HighDegree},
or from a \textsc{HighDegree} algorithm that breaks ties by \textsc{Balance}. For example, {\sc WeightedAssignment} prefers a server with load $1$ and degree $5$ to a server with load $0$ and degree $1$, whom it then prefers to a server with load $3$ and degree $6$. 

%In Appendix~\ref{app:secalg}, Figure~\ref{fig:Vexample}, we give the function $V$ for a small example with $k=d=2$ and $b=4$, along with some explanation.
%It shows that {\sc WeightedAssignment} is different from a \textsc{Balance} algorithm that breaks ties by \textsc{HighDegree}, or from a 
%\textsc{HighDegree} algorithm that breaks ties by \textsc{Balance}. 

\section{Upper bounds}\label{sec:ub}
We  will show that \textsc{WeightedAssignment} is optimal for $(k,d)$-graphs with $k\geq d$, i.e.\ no deterministic online algorithm can achieve a competitive ratio better than $c^*$. We start by proving this for the online $b$-matching problem with uniform server capacities, and later extend it to the more general problems. 

First, we show that any ($k,d$)-graph with uniform server capacities $b$ has a perfect $b$-matching, i.e. a matching where every server $s\in S$ is matched exactly $b$ times, if $k\geq d$. This generalizes Lemma~6.1 in \cite{NW}, which states this for $b=1$. 

\begin{lemma}
\label{lem:OPT}
Every ($k,d$)-graph $G=(S\cup R, E)$, where $k\geq d$, with uniform server capacities $b$ has a perfect $b$-matching.
\end{lemma}

\begin{proof}
We perform the known reduction of the $b$-matching problem to the classical maximum matching problem, i.e. we replace every server with capacity $b$ by $b$ unit servers having the same adjacent requests as the replaced server. This will result in a $(kb,db)$-graph without capacities, where $kb\geq db$. Lemma~6.1 in \cite{NW} then implies that there exists a perfect matching in this graph, which in turn implies the existence of a perfect $b$-matching in the original graph. 
\end{proof}

We move on to describing the adversary input. We start by following the construction of the previously known upper bound detailed in \cite{NW}. There are $N=d^{kb}$ servers, and the requests arrive in $kb$ rounds. Let $S_i$ denote the set of unmatched servers at the beginning of round $i$, $0\leq i< kb$. It will hold that every server in $S_i$ has a current degree of $i$ and that $|S_i|=N(1-1/d)^i$. Note that this number is always a multiple of $d$, by choice of $N$. During round $i$, $|S_i|/d$ requests are introduced, such that every request is adjacent to exactly $d$ distinct servers in $S_i$ and every server in $S_i$ gains one new neighboring request. If the online algorithm decides not to match an introduced request, we consider it matched to an arbitrary neighbor. This will only improve the performance of said algorithm. This means that a $(1-1/d)$ fraction of the servers in $S_i$ are still unmatched after round $i$, explaining the previously mentioned $S_i=N(1-1/d)^i$. Thus, there are $N\cdot (1-1/d)^{kb}$ servers with degree $kb$ and load $0$ after round $kb-1$, irrespective of what choices the algorithm makes. In a final round, further requests are introduced to all matched servers arbitrarily, such that we get a valid $(k,d)$-graph. The online algorithm will have matched at most $bN\cdot \left(1-(1-1/d)^{kb}\right)$ requests, while Lemma~\ref{lem:OPT} implies that an optimal offline algorithm matches $bN$ requests. This yields the previously known upper bound of $\left(1-(1-1/d)^{kb}\right)$.

However, it seems suboptimal to introduce requests arbitrarily after the initial $kb$ rounds. In fact, we will show that we can further limit the number of requests matched by the online algorithm if we introduce the requests more carefully. Note that all the servers that are matched during round $i$ of the previous input are similar in the sense that they all have degree $i$ and load $1$. We will apply the ideas above recursively to the sets of matched servers of all rounds. More precisely, let $T$ denote the set of matched server during some round. Say all servers in $T$ have load $l$, $0\leq l < b$, and degree $\delta$, $l\leq \delta< kb$. We then schedule $kb-\delta$ rounds for $T$ using the same construction as above. Let $T_j$ denote the set of servers that still have load $l$ at the beginning of round $j$, $0\leq j<kb-\delta$. It will now hold that every server in $T_j$ has a current degree of $\delta+j$ and that $|T_i|=|T|(1-1/d)^i$. We have to make sure that all the possible values of $|T_i|$ are multiples of $d$. This is done by increasing the initial number of servers $N$ adequately. After these $kb-\delta$ round, we have $|T|\cdot(1-1/d)^{kb-\delta}$ servers with degree $kb$ and load $l$. This process is repeated for all the sets of matched servers until we eventually obtain a valid $(k,d)$-graph, in which every server has degree $kb$. 

For this, we formally define a function $F$, where $F(x,l,\delta)$ denotes how many units of capacity we can force a deterministic online algorithm to leave empty when starting with $x$ servers that all have load $l$ and degree $\delta$. This allows us to upper bound the number of matched requests by $bN-F(N,0,0)$, yielding the following upper bound
\begin{equation}
\label{equ:cupper}
    c\leq \frac{bN-F(N,0,0)}{bN} = 1- \frac{F(N,0,0)}{bN} \,.
\end{equation}

We cannot create any empty spots on full servers, so we have $F(x,b,\delta)=0$, for all $b\leq \delta\leq kb$. Once a server has $kb$ adjacent requests, we have satisfied the $(k,d)$-graph property locally, so we do not need to introduce any more adjacent requests for this server. This is captured by $F(x,l,kb)=x\cdot(b-l)$, for all $0\leq l \leq b$. For all other combinations of $l$ and $\delta$, it is possible to define $F$ recursively. Recall that we introduce $kb-\delta$ rounds when starting with $x$ servers all with load $l$ and degree $\delta$. During each of these rounds, exactly a $1/d$ fraction of the servers that still had load $l$ at the start of the round are discarded. Moreover, every server gets exactly one new neighbor during each round until they are matched. This implies 
\begin{equation}
    \label{equ:Frecursive}
    F(x,l,\delta) = x\cdot\left(1-\frac{1}{d}\right)^{kb-\delta}(b-l) + \sum_{i=1}^{kb-\delta} F\left(x\cdot\frac{1}{d} \left(1-\frac{1}{d} \right)^{i-1},l+1,\delta+i  \right) \,, 
\end{equation}
for all $0\leq l<b$ and $l \leq \delta <kb$. The following lemma solves this recurrence, which we can then apply in (\ref{equ:cupper}) to obtain the theorem. 

\begin{lemma}
\label{lem:F}
For all $l$, $0\leq l\leq b$, and all $\delta$, $l\leq \delta \leq kb$, it holds that 
\begin{equation}
     \label{equ:Fsolved}
    F(x,l,\delta) = x \left(1-\frac{1}{d}\right)^{kb-\delta} \left(\sum_{i=1}^{b-l}i \binom{kb-\delta}{b-l-i} \frac{1}{(d-1)^{b-l-i}}\right)\,. 
\end{equation}
\end{lemma}

\begin{proof}
By induction over $l$, starting with $l=b$ and going down to $l=0$. The induction base is satisfied as we get the empty sum in (\ref{equ:Fsolved}), making the whole expression zero. In the induction step $l+1\leadsto l$, we can apply our induction hypothesis to $F$ in (\ref{equ:Frecursive}). 
\begin{align*}
    &\sum_{i=1}^{kb-\delta} F\left(x\cdot \frac{1}{d} \left(1-\frac{1}{d} \right)^{i-1},l+1,\delta+i  \right) \\ 
    &\overset{\text{IH}}{=} \sum_{i=1}^{kb-\delta} x\frac{1}{d} \left(1-\frac{1}{d} \right)^{i-1} \left(1-\frac{1}{d}\right)^{kb-(\delta+i)} \left(\sum_{j=1}^{b-(l+1)}j \binom{kb-\delta-i}{b-(l+1)-j} \frac{1}{(d-1)^{b-(l+1)-j}}\right) \\ 
    &=x\frac{1}{d}\left(1-\frac{1}{d}\right)^{kb-\delta-1}  \left(\sum_{j=1}^{b-(l+1)}j \frac{1}{(d-1)^{b-(l+1)-j}} \sum_{i=1}^{kb-\delta}\binom{kb-\delta-i}{b-(l+1)-j}\right) \\ 
    &= x\left(1-\frac{1}{d}\right)^{kb-\delta}\frac{1}{d-1} \left( \sum_{j=1}^{b-(l+1)}j \frac{1}{(d-1)^{b-l-j-1}} \sum_{i=0}^{kb-\delta-1}\binom{i}{b-l-j-1}\right)\\ 
    &= x\left(1-\frac{1}{d}\right)^{kb-\delta}\left( \sum_{j=1}^{b-(l+1)}j \frac{1}{(d-1)^{b-l-j}} \binom{kb-\delta}{b-l-j}\right) \,,
\end{align*}
where we used Fact~\ref{fac:binom} in the last step. Finally, we can plug this result back into (\ref{equ:Frecursive}) to obtain 
\begin{align*}
     F(x,l,\delta) &= x\cdot\left(1-\frac{1}{d}\right)^{kb-\delta}(b-l) + \sum_{i=1}^{kb-\delta} F\left(x\cdot\frac{1}{d} \left(1-\frac{1}{d} \right)^{i-1},l+1,\delta+i  \right) \\
      &= x\cdot\left(1-\frac{1}{d}\right)^{kb-\delta}\left((b-l) + \sum_{j=1}^{b-(l+1)}j \frac{1}{(d-1)^{b-l-j}} \binom{kb-\delta}{b-l-j}\right) \\
      &= x\cdot\left(1-\frac{1}{d}\right)^{kb-\delta}\left(\sum_{j=1}^{b-l}j \frac{1}{(d-1)^{b-l-j}} \binom{kb-\delta}{b-l-j}\right) \,. \qedhere
\end{align*}
\end{proof}

\begin{theorem}\label{th:ub}
No deterministic online algorithm for the $b$-matching problem with uniform server capacities $b$ can achieve a competitiveness better than $c^*$ on $(k,d)$-graphs with $k\geq d$.
\end{theorem}

We extend this upper bound to the more general case with variable server capacities, which we will examine in the next section. Let $b_{min} = \min_{s\in S} b_s$.
The optimal competitiveness is then 
\[
    c^*_{\min} = 1-\frac{1}{b_{\min}}\left(\sum_{i=1}^{b_{\min}} i \binom{kb_{\min}}{b_{\min}-i}\frac{1}{(d-1)^{b_{\min}-i}}\right)\left(1-\frac{1}{d} \right)^{kb_{\min}} \,. 
\]
%where $b_{min} = \min_{s\in S} b_s$. The proof the next corollary is given in Appendix~\ref{app:proof}. 
\begin{corollary}
\label{cor:varCap}
No deterministic online algorithm for the $b$-matching problem with variable server capacities can achieve a competitive ratio better than $c^*_{\min}$ on $(k,d)$-graphs with $k\geq d$.
\end{corollary}

\begin{proof}
Note that we can upscale the number of initial servers $N$ and thus the number of matched requests in the adversary input for uniform server capacities. This means that for any set of server capacities $B$ (with $|B|<\infty$), we can create the adversary input with $N$ servers for the capacity $b'$ which minimizes the corresponding competitive ratio $c^*$. Moreover, we only create one server for every capacity in $B\setminus\{b'\}$ and add enough neighbors that are only adjacent to the respective server such that we obtain a valid $(k,d)$-graph. Even though this will increase the number of matched requests by $\sum_{b\in B\setminus\{b'\}} b$, this increase can be made null by increasing $N$ sufficiently. In the limit, any deterministic online algorithm achieves a competitive ratio of $\min_{s\in S}c^*_s$. In Section~\ref{sec:comp}, we show that $c^*$ is strictly increasing in $b$ if $k\geq d$, implying that $\min_{s\in S}c^*_s=c^*_{\min}$.
\end{proof}

%The proof is given in Appendix~\ref{app:secub}. 
Obviously, Theorem~\ref{th:ub} and Corollary~\ref{cor:varCap} also hold for the more general vertex-weighted $b$-matching problem, addressed in
the next section. 
% Moreover, the vertex-weighted problem reduces to the unweighted problem when all servers have equal weight, meaning that the same upper bounds hold.
% \begin{corollary}
% No deterministic online algorithm for the vertex-weighted $b$-matching problem with variable server capacities can achieve a competitive ratio better than $c^*_{\min}$ on % $(k,d)$-graphs with $k\geq d$.
% \end{corollary}

\section{Variable server capacities and vertex weights}\label{sec:ext}
We detail the necessary changes to \textsc{WeightedAssignment} such that it can handle variable server capacities as well as vertex weights, while still achieving the optimal competitive ratio. 

%\subsection{Variable server capacities}

\noindent {\bf Variable server capacities:} 
Recall that every server $s\in S$ now has a server capacity $b_s$, and thus a degree of at least $d(s)\geq k\cdot b_s$. This changes the objective function of the dual linear program to 
\[
    \sum_{s\in S} b_s \cdot x(s) + \sum_{r\in R} y(r) \,.
\]
We handle this by computing the function $V_s$ for every server individually, for its capacity $b_s$. This means that we construct $V_s$ according to Section~\ref{sec:alg}, such that  
\begin{equation}
\label{equ:Vvar}
    b_s\cdot \Big( V_s(l+1,\delta+1) - V_s(l,\delta) + (d-1)\cdot \big(V_s(l,\delta+1) - V_s(l,\delta)\big)\Big) \leq \frac{1}{c^*_s} \,, 
\end{equation}
where $c^*_s$ is equal to $c^*$, but $b$ is replaced by $b_s$. Moreover, we have $V_s(b_s,\cdot)=1$ and $V_s(\cdot, \delta')=1$ if $\delta'\geq kb_s$, meaning that we again construct a feasible dual solution. However, we still have to adapt the decision criterion of \textsc{WeightedAssignment}. We change Line~7 in Algorithm~\ref{alg:WA} to
\[
    \text{Match $r$ to } \arg \max\left\{b_s\cdot \big(V_s(l_s,\delta_s+1)-V_s(l_s,\delta_s)\big): s\in N(r)\right\}\,.
\]
This allows us to upper bound the total increase in the dual solution when the adapted strategy, called \textsc{WeightedAssignment(VC)}, assigns a request $r$ to a server $s$ by 
\begin{align*}
    \Delta D &= b_s \cdot \Big( V_s(l_s+1,\delta_{s}+1) - V_s(l_s,\delta_{s})\Big) + \sum_{s'\in N(r)\setminus\{s\}} b_{s'} \cdot \Big(V_{s'}(l_{s'},\delta_{s'}+1) - V_{s'}(l_{s'},\delta_{s'})\Big) \\ 
    &\leq  b_s\cdot \Big( V_s(l_s+1,\delta_{s}+1) - V_s(l_s,\delta_{s}) + (d-1)  \cdot \big(V_{s}(l_{s},\delta_{s}+1) - V_{s}(l_{s},\delta_{s})\big)\Big) \leq \frac{1}{c^*_s} \,.
\end{align*}
Thus, \textsc{WeightedAssignment(VC)} achieves a competitive ratio of $\min_{s\in S} c^*_s$. In Section~\ref{sec:comp} we will show that $c^*$ is monotonically increasing in $b$ for $k\geq d$, meaning that $\min_{s\in S} c^*_s=c^*_{\min}$, cf.\  Section~\ref{sec:ub}.

\begin{theorem}
\textsc{WeightedAssignment(VC)} achieves a competitive ratio of $\min_{s\in S} c^*_s$ the $b$-matching problem with variable server capacities on $(k,d)$-graphs. 
The ratio equals $c^*_{\min}$ and is optimal for $k\geq d$. 
\end{theorem}

%\subsection{Vertex weights}
\smallskip

\noindent {\bf Vertex weights:} 
At last, we consider the vertex-weighted extension of the online $b$-matching problem. Every server $s\in S$ now has a weight $w_s$ assigned to it, and the value of every matching edge incident to $s$ is multiplied by $w_s$. %The goal is to construct a matching of maximum weight. 
This problem is modelled by the following linear programs. 
\begin{align*}
\textbf{P: }\text{max} \ &\sum_{\{s,r\}\in E} w_s \cdot m(s,r) 
& \textbf{D: }\text{min} \ &\sum_{s\in S} w_s\cdot b_s \cdot x(s) + \sum_{r\in R} y(r) \\
\text{s.t.} \ &\sum_{r:\{s,r\}\in E} w_s\cdot m(s,r) \leq w_s\cdot b_s, \ (\forall s \in S)
&\text{s.t.} \ & w_s\cdot x(s)+y(r) \geq w_s, \ (\forall \{s,r\}\in E) \\
&\sum_{s:\{s,r\}\in E} m(s,r) \leq 1, \ (\forall r\in R)
& &x(s),\,y(r)\geq 0, \ (\forall s\in S, \forall r\in R) \\
& m(s,r) \geq 0, \ (\forall \{s,r\}\in E)
\end{align*}
%\begin{align*}
%\textbf{Primal: }\text{max} \ &\sum_{\{s,r\}\in E} w_s \cdot m(s,r) && \\
%\text{s.t.} \ &\sum_{r:\{s,r\}\in E} w_s\cdot m(s,r) \leq w_s\cdot b_s, &&\hspace{-21mm} (\forall s \in S) \\
%&\sum_{s:\{s,r\}\in E} m(s,r) \leq 1, &&\hspace{-21mm} (\forall r\in R) \\
%& m(s,r) \geq 0, &&\hspace{-21mm} (\forall \{s,r\}\in E) \, .
%\end{align*}
%\begin{align*}
%\rlap{\textbf{Dual:}}\phantom{\textbf{Primal: }}\text{min} \ &\sum_{s\in S} w_s\cdot b_s \cdot x(s) + \sum_{r\in R} y(r) && \\
%\text{s.t.} \ & w_s\cdot x(s)+y(r) \geq w_s, &&\hspace{-10mm} (\forall \{s,r\}\in E) \\
%& x(s),y(r) \geq 0, &&\hspace{-10mm} (\forall s\in S, \forall r\in R) \, .
%\end{align*}

Observe that $x(s)=1$, for all $s\in S$, by the end of the algorithm still implies dual feasibility. Hence, we do not need to change the construction of $V_s$, and still obtain a feasible dual solution. All we need to do is to change the decision criterion of \textsc{WeightedAssignment(VC)} once more to 
\[
    \text{Match $r$ to } \arg \max\left\{w_s \cdot b_s\cdot \big(V_s(l_s,\delta_s+1)-V_s(l_s,\delta_s)\big): s\in N(r)\right\}\,.
\]
Whenever the resulting algorithm \textsc{WeightedAssignment(VW)} assigns a request $r$ to a server $s$, we increase the primal solution by $w_s$, while we can upper bound the increase in the dual solution by 
\begin{align*}
    \Delta D &\leq  w_s \cdot b_s\cdot \Big( V_s(l_s+1,\delta_{s}+1) - V_s(l_s,\delta_{s}) + (d-1)  \cdot \big(V_{s}(l_{s},\delta_{s}+1) - V_{s}(l_{s},\delta_{s})\big)\Big) \leq \frac{w_s}{c^*_s}\,.
\end{align*}
This again yields a competitiveness of $\min_{s\in S} c^*_s$.  %c^*_{\min}$.

\begin{theorem}
\textsc{WeightedAssignment(VW)} achieves a competitive ratio of $\min_{s\in S} c^*_s$ for the
vertex-weighted $b$-matching problem with variable server capacities on $(k,d)$-graphs. The ratio equals $c^*_{\min}$
and is optimal for $k\geq d$. 
\end{theorem}

\section{Analysis of the competitive ratio}\label{sec:comp} 

\begin{theorem}
\label{thm:conv}
If $k\geq d\geq2$, the competitive ratio $c^*$ converges to one as $b$ tends to infinity, that is $\lim_{b\rightarrow\infty}c^*=1$.
\end{theorem}

\begin{proof}
We show that 
\[
    \lim_{b\rightarrow \infty} \frac{1}{b}\left(\sum_{i=1}^b i \binom{kb}{b-i}\frac{1}{(d-1)^{b-i}}\right)\left(1-\frac{1}{d} \right)^{kb} = 0 \,, 
\]
for all $k,d\in \mathbb{N}$ with $k\geq d\geq 2$. As a first step, we prove
\begin{equation}
    \label{equ:step1}
    \frac{1}{b}\left(\sum_{i=1}^b i \binom{kb}{b-i}\frac{1}{(d-1)^{b-i}}\right) \leq \binom{kb}{b-1} \frac{1}{(d-1)^{b-1}} \, .
\end{equation}
For this, observe that 
\[
    \binom{n}{k} = \binom{n}{k-i} \prod_{j=1}^i \frac{n-k+j}{k+1-j} \,,
\]
which means
\begin{align}
\label{equ:extendedBinom}
    \sum_{i=1}^b i \binom{kb}{b-i}\frac{1}{(d-1)^{b-i}} 
    %&=\sum_{i=1}^b i \prod_{j=1}^{i-1} \frac{b-j}{kb-(b-1)+j}\binom{kb}{b-1}\frac{1}{(d-1)^{b-1}} (d-1)^{i-1} \\ 
    &= \binom{kb}{b-1}\frac{1}{(d-1)^{b-1}}\sum_{i=1}^b i  (d-1)^{i-1} \prod_{j=1}^{i-1} \frac{b-j}{kb-(b-1)+j} \,.
\end{align}
In order to accurately bound the sum above, we rely on the theory of \emph{hypergeometric functions}. 
\begin{definition}
A Gauss hypergeometric function ${}_2F_1(\alpha,\beta;\gamma;z)$ can be defined as 
\[
    {}_2F_1(\alpha,\beta;\gamma;z)=1+\frac{\alpha\cdot \beta}{\gamma \cdot 1} z + \frac{\alpha(\alpha+1)\cdot \beta(\beta+1)}{\gamma(\gamma+1) \cdot 1\cdot 2} z^2 + \ldots = \sum_{n=0}^\infty \frac{(\alpha)_n (\beta)_n}{(\gamma)_n}\frac{z^n}{n!} \, ,
\]
where $\alpha,\beta,\gamma$ are parameters and $z$ is the variable. Here, $(x)_n$ denotes the Pochhammer symbol, which is defined as $(x)_n = \prod_{i=0}^{n-1} (x+i)$.
\end{definition}

Two hypergeometric functions are called \emph{contiguous} if they have the same variable and two equal parameters, while the third parameter differs by an integer. It is known \cite[Chapter~2.5]{andrews_askey_roy_1999} that for any three contiguous Gauss hypergeometric functions, there is a linear relation with rational coefficients depending on the parameters and the variable. We use the following relation 
\begin{equation}
    \label{equ:hyper}
    \alpha(1-z) {}{}_2F_1(\alpha+1,\beta;\gamma;z)+(\gamma -2\alpha - (\beta-\alpha)z) {}_2F_1(\alpha,\beta;\gamma;z) - (\gamma - \alpha) {}_2F_1(\alpha-1,\beta;\gamma;z)=0 \,.
\end{equation}

We upper bound the sum in (\ref{equ:extendedBinom}) by a hypergeometric function and then use (\ref{equ:hyper}) to determine its value. For this, observe that if one of the parameters $\alpha$ or $\beta$ of a hypergeometric function is a negative integer, it simplifies to a polynomial. For example, let $\beta=-m$ where $m\in \mathbb{N}$. Then, we get
\[
    {}_2F_1(\alpha,-m,\gamma,z) = \sum_{n=0}^m (-1)^n \frac{m!}{(m-n)!} \frac{(\alpha)_n}{(\gamma)_n}\frac{z^n}{n!} = \sum_{n=0}^m (-1)^n \binom{m}{n} \frac{(\alpha)_n}{(\gamma)_n} z^n \,.
\]
Furthermore, we have 
\begin{align*}
    \sum_{i=1}^b i  (d-1)^{i-1} \prod_{j=1}^{i-1} \frac{b-j}{kb-(b-1)+j} &\leq \sum_{i=1}^b i  (d-1)^{i-1} \prod_{j=1}^{i-1} \frac{b-j}{db-(b-1)+j} \\ 
    &= \sum_{i=0}^{b-1} (i+1)  (d-1)^{i} \prod_{j=0}^{i-1} \frac{b-1-j}{(d-1)b+2+j} \\
    &= \sum_{i=0}^{b-1} \frac{(i+1)!}{i!} (-1)^i (1-d)^{i} \frac{(b-1)!}{(b-1-i)!} \frac{1}{\big((d-1)b+2 \big)_i} \\
    &= \sum_{i=0}^{b-1} (-1)^i \binom{b-1}{i} \frac{(2)_i}{\big((d-1)b+2 \big)_i} (1-d)^i \\
    &= {}_2F_1(2,1-b;(d-1)b+2;1-d) \,.
\end{align*}
Using (\ref{equ:hyper}) with $\alpha=1$, $\beta=1-b$, $\gamma=(d-1)b+2$ and $z=1-d$ yields the coefficients $\alpha(1-z)=d$, $\gamma -2\alpha - (\beta-\alpha)z = 0$ and $\gamma - \alpha= db - (b-1)$.
%\begin{align*}
%    \alpha(1-z) &= d \\ 
%    \gamma -2\alpha - (\beta-\alpha)z &= 0 \\
%    \gamma - \alpha&= db - (b-1) \,.
%\end{align*}
Moreover, it holds that $F(0,\beta;\gamma;z)=1$. Hence, we get 
\[
    {}_2F_1(2,1-b;(d-1)b+2;1-d) = b - \frac{b-1}{d} \leq b \,,
\]
implying (\ref{equ:step1}). All that is left to show now is 
\[
    \lim_{b\rightarrow \infty}\binom{kb}{b-1} \frac{1}{(d-1)^{b-1}}\left(1-\frac{1}{d} \right)^{kb} = 0 \,.
\]

Here, we will use the Stirling approximation for $\binom{kb}{b-1}$,
%\[ 
%    \binom{n}{k} \sim \sqrt{\frac{n}{2\pi k(n-k)}}\cdot \frac{n^n}{k^k (n-k)^{n-k}} \,,
%\]
which yields
\begin{align*} 
    \binom{kb}{b-1} \frac{1}{(d-1)^{b-1}}\left(1-\frac{1}{d} \right)^{kb} \sim \, &\sqrt{\frac{kb}{2\pi (b-1)(kb-b+1)}} \\ &\cdot \frac{(kb)^{kb}}{(b-1)^{b-1} (kb-b+1)^{kb-b+1}} \cdot \frac{(d-1)^{kb-b+1}}{d^{kb}}\,.
\end{align*}
First, notice that the root converges to $0$ if $b\rightarrow\infty$,
%\[
%    \lim_{b\rightarrow \infty}\sqrt{\frac{kb}{2\pi (b-1)(kb-b+1)}} = 0 \,,
%\]
which means that we only have to upper bound the remaining terms by a constant in order to finish the proof. We have 
\begin{align*}
    \frac{(kb)^{kb}}{(b-1)^{b-1} (kb-b+1)^{kb-b+1}} \cdot \frac{(d-1)^{kb-b+1}}{d^{kb}} &\leq \frac{1}{(b-1)^{b-1}}  \cdot \frac{(kb)^{kb}}{d^{kb}} \cdot  \frac{(d-1)^{kb-b+1}}{((k-1)b)^{kb-b+1}} \\ 
    &= \left(\frac{b}{b-1} \right)^{b-1}  \left(\frac{k}{d}\right)^{kb}\left(\frac{d-1}{k-1} \right)^{kb-b+1} \,.
\end{align*}
Note that the first term is upper bounded by $e$. Moreover, with the help of the following lemma, we can upper bound the remaining terms by 1. 
\begin{lemma}
\label{lem:conv}
For all $k\geq 2$ and $d\geq 2$, it holds that 
\[
\frac{k^k}{(k-1)^{k-1}} \leq \frac{d^k}{(d-1)^{k-1}} \,.
\]
\end{lemma}
\begin{proof}
Consider the function $f(x)={x^k}/{(x-1)^{k-1}}$ for $x>1$. We show that $x=k$ is a global minimum. It holds
%\[
%    f(x)=\frac{x^k}{(x-1)^{k-1}}
%\]
\[
    f'(x) = \frac{x^{k-1}(x-k)}{(x-1)^{k}}\,.
\]
For $x>1$, we only have $f'(x)= 0$ if $x=k$. Furthermore, it holds 
\[
    f''(x) = \frac{(k-1)kx^{k-2}}{(x-1)^{k+1}}\,.
\]
It is easy to see that $f''(x)>0$ for all $k\geq 2$ and $x>1$. Moreover, since $\lim_{x\rightarrow 1+} f(x) = \infty$ and $\lim_{x\rightarrow \infty} f(x) = \infty$, we indeed have a global minimum at $x=k$. This finishes the proof. 
\end{proof}

Lemma~\ref{lem:conv} implies 
\[
\left(\frac{k}{d} \right)^{k} \leq \left(\frac{k-1}{d-1} \right)^{k-1}\,,
\]
and since $x^b$ is strictly increasing in $x$ for $x>0$ and $b>0$, we have 
\[ 
    \left(\frac{k}{d}\right)^{kb}\left(\frac{d-1}{k-1} \right)^{kb-b+1} \leq \left(\frac{k-1}{d-1} \right)^{(k-1)b}\left(\frac{d-1}{k-1} \right)^{(k-1)b} \cdot \frac{d-1}{k-1} = \frac{d-1}{k-1} \leq 1\,,
\]
as $k\geq d\geq 2$. 
\end{proof}

\begin{theorem}
\label{thm:mono}
If $k\geq d\geq2$, the competitive ratio $c^*$ is strictly increasing in $b$, for $b\geq 1$. 
\end{theorem}

\begin{proof}
We prove that $c^*$ is strictly increasing in $b$ by showing
\begin{align*}
    c'_b:= &\frac{1}{b}\left(\sum_{i=1}^b i \binom{kb}{b-i}\frac{1}{(d-1)^{b-i}}\right)\left(1-\frac{1}{d} \right)^{kb} > \\ 
    &\frac{1}{b+1}\left(\sum_{i=1}^{b+1} i \binom{k(b+1)}{b+1-i}\frac{1}{(d-1)^{b+1-i}}\right)\left(1-\frac{1}{d} \right)^{k(b+1)} =: c'_{b+1}\, ,
\end{align*}
for all $b\geq 1$, if $k\geq d\geq 2$. We do this by showing that $c'_{b+1}/c'_b< 1$. As detailed in the proof of Theorem~\ref{thm:conv}, it holds that 
\[
    %c'_b= \frac{1}{b}\left(\sum_{i=0}^{b-1} (i+1)  (d-1)^{i} \prod_{j=0}^{i-1} \frac{b-1-j}{(k-1)b+2+j} \right) \binom{kb}{b-1}\frac{1}{(d-1)^{b-1}}\left(1-\frac{1}{d} \right)^{kb}\, .
    c'_b= \frac{{}_2F_1(2,1-b;(k-1)b+2;1-d)}{b} \binom{kb}{b-1}\frac{1}{(d-1)^{b-1}}\left(1-\frac{1}{d} \right)^{kb}\, .
\]
It follows that 
\[
    \frac{c'_{b+1}}{c'_b} = \frac{b \cdot {}_2F_1(2,-b;(k-1)(b+1)+2;1-d)}{(b+1)\cdot{}_2F_1(2,1-b;(k-1)b+2;1-d)} \cdot  \frac{\binom{k(b+1)}{b}}{\binom{kb}{b-1}} \cdot \frac{(d-1)^{k-1}}{d^{k}} \,.
\]
%The following two lemmas bound the fractions above. 

\begin{lemma}
\label{lem:fracHyper}
For all $b\geq 1$ and $k\geq d\geq 2$, it holds that 
\[
\frac{b \cdot {}_2F_1(2,-b;(k-1)(b+1)+2;1-d)}{(b+1)\cdot{}_2F_1(2,1-b;(k-1)b+2;1-d)} < 1 \,.
\]
\end{lemma}

\begin{proof}
In the following, we fix $k\geq 2$ and $b\geq 1$. Then, we show that 
\[
\frac{{}_2F_1(2,-b;(k-1)(b+1)+2;1-d)}{{}_2F_1(2,1-b;(k-1)b+2;1-d)}  
\]
is maximized if $d$ is maximized, i.e. $d=k$. As shown in the proof of Theorem~\ref{thm:conv}, we will then have for $d=k$
\[
\frac{\frac{1}{b+1}{}_2F_1(2,-b;(k-1)(b+1)+2;1-k)}{\frac{1}{b}{}_2F_1(2,1-b;(k-1)b+2;1-k)} = \frac{ \frac{1}{b+1} \left(b+1-\frac{b}{k}\right)}{\frac{1}{b}\left(b-\frac{b-1}{k}\right)} = \frac{1-\frac{1}{k}\frac{b}{b+1}}{1-\frac{1}{k}\frac{b-1}{b}} < 1 \,,
\]
since $b/(b+1) > (b-1)/b$ for all $b\geq 1$. Recall, that 
\begin{align*}
    {}_2F_1(2,-b;(k-1)(b+1)+2;1-d) &= \sum_{i=0}^b (i+1)(d-1)^i \frac{b!}{(b-i)!}\frac{1}{\left((k-1)(b+1)+2 \right)_i} \\ 
    &=: \sum_{i=0}^b u_i\cdot (d-1)^i \,,
\end{align*}
  and 
\begin{align*}
    {}_2F_1(2,1-b;(k-1)b+2;1-d) &= \sum_{i=0}^{b-1} (i+1)(d-1)^i \frac{(b-1)!}{(b-1-i)!}\frac{1}{\left((k-1)b+2 \right)_i} \\
    &=: \sum_{i=0}^{b-1} v_i\cdot (d-1)^i\,,
\end{align*}
where $u_i$ and $v_i$ are constants, since we fixed $k$ and $b$. This means, that the fraction of the two hypergeometric functions is a quotient of two real polynomials. Thus, we can use the following theorem given in \cite{HVV}.

\begin{theorem}[{{\cite[Theorem 4.4]{HVV}}}]
Let $f_n(x)=\sum_{k=0}^n a_k x^k$ and $g_n(x)=\sum_{k=0}^n b_k x^k$ be real polynomials, with $b_k>0$ for all $k$. If the sequence $\{a_k/b_k\}$ is increasing (decreasing), then so is the function $f_n(x)/g_n(x)$ for all $x>0$. 
\end{theorem}

First, it is easy to see that $u_i>0$, for all $0\leq i \leq b$, and $v_i>0$, for all $0\leq i \leq b-1$, hold. Thus, we can apply the theorem above to the fraction 
\[
    \frac{f(d)}{g(d)}:= \frac{\sum_{i=0}^{b-1} u_i (d-1)^i}{\sum_{i=0}^{b-1} v_i (d-1)^i}\,.
\]
We will show that $(u_i/v_i)_{0\leq i\leq b-1}$ forms an increasing sequence, implying that $f/g$ is increasing for all $d>1$. Moreover, we have that
\[
    \frac{u_b(d-1)^b}{\sum_{i=0}^{b-1} v_i (d-1)^i}
\]
is increasing for $d> 1$, because the reciprocal
\[
    \sum_{i=0}^{b-1} \frac{v_i}{u_b}\cdot \frac{1}{(d-1)^{b-i}}
\]
is decreasing as all the coefficients are greater than zero. It then immediately follows that the considered fraction of hypergeometric functions is increasing for $d>1$. 

Hence, all that is left to prove is that $\left(u_i/v_i \right)_{0\leq i\leq b-1}$ is indeed an increasing sequence, i.e. $u_i/v_i > u_{i-1}/v_{i-1}$,
%\begin{equation*}
%    u_i/v_i > u_{i-1}/v_{i-1}\,,
%\end{equation*}
for all $1\leq i \leq b-1$. We have
\begin{align*}
    \frac{u_i}{v_i} = \frac{b}{b-i} \frac{((k-1)b+2)_i}{((k-1)(b+1)+2)_i} \quad \text{and} \quad
    \frac{u_{i-1}}{v_{i-1}} = \frac{b}{b-i+1} \frac{((k-1)b+2)_{i-1}}{((k-1)(b+1)+2)_{i-1}} \,. 
\end{align*} 
Observe that 
\[
    \frac{u_i}{v_i} = \frac{u_{i-1}}{v_{i-1}} \cdot \frac{b-i+1}{b-i} \frac{(k-1)b+2+i-1}{(k-1)(b+1)+2+i-1}\,,
\]
meaning that we only have to show 
\begin{equation*}
    \frac{(k-1)b+2+i-1}{(k-1)(b+1)+2+i-1} > \frac{b-i}{b-i+1} \,.
\end{equation*}
The left fraction is increasing in $i$, and thus lower bounded by 
\[
    \frac{(k-1)b+1}{(k-1)(b+1)+1} > \frac{b}{b+1} \,.
\]
On the other hand, $(b-i)/(b-i+1)$ is decreasing in $i$, and thus upper bounded by $b/(b+1)$. This finishes the proof. 
\end{proof}

\begin{lemma}
\label{lem:fracBinom}
For all $k\geq 2$ and all $b\geq 1$, it holds that 
\[
\frac{\binom{k(b+1)}{b}}{\binom{kb}{b-1}} \leq \frac{k^k}{(k-1)^{k-1}} \,.
\]
\end{lemma}

\begin{proof}
We prove this by considering the fraction $\binom{k(b+1)}{b+1}/\binom{kb}{b}$, which is closely related to the desired fraction of binomial coefficients in the sense that
\begin{equation}
\label{equ:binom}
    \frac{\binom{k(b+1)}{b}}{\binom{kb}{b-1}} = \frac{b+1}{b}\frac{(k-1)b+1}{(k-1)(b+1)+1} \frac{\binom{k(b+1)}{b+1}}{\binom{kb}{b}} \,.
\end{equation}
We upper bound said fraction of binomial coefficients by applying upper and lower bounds for $n!$. The Stirling approximation yields 
\[
    \sqrt{2 \pi } n^{n+\frac{1}{2}} e^{-n} e^{\frac{1}{12n+1}} < n! < \sqrt{2 \pi } n^{n+\frac{1}{2}} e^{-n} e^{\frac{1}{12n}} \,.
\]
We can rewrite the fraction of binomial coefficients as
\[
    \frac{\binom{k(b+1)}{b+1}}{\binom{kb}{b}} = \frac{1}{b+1} \frac{\big(k(b+1) \big)!}{\big((k-1)(b+1) \big)!}\frac{\big( (k-1)b\big)!}{(kb)!} \,.
\]
Now, we apply the upper bound in the numerator and the lower bound in the denominator 
\begin{align*}
    \frac{\binom{k(b+1)}{b+1}}{\binom{kb}{b}} < \frac{1}{b+1}& \frac{\big(k(b+1) \big)^{k(b+1)+\frac{1}{2}}e^{-k(b+1)}}{\big( (k-1)(b+1)\big)^{(k-1)(b+1)+\frac{1}{2}} e^{-(k-1)(b+1)}} \cdot e^{\frac{1}{12k(b+1)}-\frac{1}{12(k-1)(b+1)+1}} \\ 
    & \frac{\big((k-1)b \big)^{(k-1)b+\frac{1}{2}}e^{-(k-1)b}}{\big( kb\big)^{kb+\frac{1}{2}} e^{-kb}} \cdot e^{\frac{1}{12(k-1)b}-\frac{1}{12kb+1}}\,.
\end{align*}
It holds that
\[
    \frac{1}{12k(b+1)}-\frac{1}{12(k-1)(b+1)+1} < 0 \quad \text{and}\quad \frac{1}{12(k-1)b}-\frac{1}{12kb+1} < \frac{1}{12(k-1)b} \,,
\]
which implies
\begin{align*}
    \frac{\binom{k(b+1)}{b+1}}{\binom{kb}{b}} &< \frac{1}{b+1} \left(\frac{k}{k-1} \right)^{(k-1)(b+1)+\frac{1}{2}} k^{b+1} (b+1)^{b+1} e^{-(b+1)}
    \\ &\phantom{<}\,\cdot
    \left(\frac{k-1}{k} \right)^{(k-1)b+\frac{1}{2}} \left( \frac{1}{k}\right)^b \left( \frac{1}{b} \right)^b e^b e^{\frac{1}{12(k-1)b}} \\ 
    &= \frac{k^k}{(k-1)^{k-1}} \left(\frac{b+1}{b} \right)^b \frac{1}{e} e^{\frac{1}{12(k-1)b}}\,.
\end{align*}
Plugging this into equation~(\ref{equ:binom}) then yields 
\begin{equation*}
     \frac{\binom{k(b+1)}{b}}{\binom{kb}{b-1}} < \frac{(k-1)b+1}{(k-1)(b+1)+1} \frac{k^k}{(k-1)^{k-1}} \left(\frac{b+1}{b} \right)^{b+1} \frac{1}{e} e^{\frac{1}{12(k-1)b}}\,.
\end{equation*}

Thus, in order to prove the desired upper bound, we have to show 
\[
    \left(\frac{b+1}{b} \right)^{b+1} \frac{(k-1)b+1}{(k-1)(b+1)+1}  e^{\frac{1}{12(k-1)b}} \leq e  \,.
\]
Notice that all terms are decreasing in $k$, for $k\geq 2$, meaning that we can upper bound the left-hand side by 
\[
    \left(\frac{b+1}{b} \right)^{b+1} \frac{(k-1)b+1}{(k-1)(b+1)+1}  e^{\frac{1}{12(k-1)b}} \leq \left(\frac{b+1}{b} \right)^{b+1} \frac{b+1}{b+2}  e^{\frac{1}{12b}} = \left(\frac{b+1}{b} \right)^{b} \frac{(b+1)^2}{b(b+2)}  e^{\frac{1}{12b}}  \,.
\]
As a next step, we take logarithms on both sides. We make use the following useful inequality, which is true for $x>0$
\[ 
    \ln\left(1+x \right) \leq \frac{x(6+x)}{6+4x}\,.
\]
We obtain 
\begin{align*}
    b \ln\left(1+\frac{1}{b} \right) + \ln\left(1+\frac{1}{b^2+2b} \right)+\frac{1}{12b} \leq \frac{6b+1}{6b+4} + \frac{1}{b^2+2b} + \frac{1}{12b } \overset{!}{\leq} 1\,,
\end{align*}
where we also used the simpler upper bound $\ln(1+x)\leq x$ for the second term. One can show that the inequality above is indeed satisfied for $b>1.66$. This means that we have to show the case $b=1$ separately. There, we have
\[ 
    \frac{\binom{2k}{1}}{\binom{k}{0}} = 2k \leq \left( \frac{k}{k-1}\right)^{k-1} k = \frac{k^k}{(k-1)^{k-1}} \,,
\]
where the first inequality is true as $k\geq 2$. 
\end{proof}

Finally, we can finish the proof of Theorem~\ref{thm:mono} with the help of Lemmas~\ref{lem:fracHyper} and \ref{lem:fracBinom}. We have
\begin{align*}
    \frac{c'_{b+1}}{c'_b} &= \frac{b \cdot {}_2F_1(2,-b;(k-1)(b+1)+2;1-d)}{(b+1)\cdot{}_2F_1(2,1-b;(k-1)b+2;1-d)} \cdot  \frac{\binom{k(b+1)}{b}}{\binom{kb}{b-1}} \cdot \frac{(d-1)^{k-1}}{d^{k}} \\
    &< 1\cdot \frac{k^k}{(k-1)^{k-1}} \cdot \frac{(d-1)^{k-1}}{d^{k}} \leq 1 \,,
\end{align*}
where the last step is true according to Lemma~\ref{lem:conv}.
\end{proof}

%%
%% Bibliography
%%

%% Please use bibtex, 
\bibliographystyle{plain}
\bibliography{paper}

\begin{thebibliography}{10}

\bibitem{AGKM}
G.\ Aggarwal, G.\ Goel, C.\ Karande, and A.\ Mehta.
\newblock Online vertex-weighted bipartite matching and single-bid budgeted
  allocations.
\newblock In {\em Proc.\ 22nd Annual {ACM-SIAM} Symposium on Discrete
  Algorithms (SODA)}, pages 1253--1264. {SIAM}, 2011.

\bibitem{andrews_askey_roy_1999}
George~E. Andrews, Richard Askey, and Ranjan Roy.
\newblock {\em Special Functions}.
\newblock Encyclopedia of Mathematics and its Applications. Cambridge
  University Press, 1999.

\bibitem{ACR}
Y.\ Azar, I.R.\ Cohen, and A.~Roytman.
\newblock Online lower bounds via duality.
\newblock In {\em Proc.\ 28th Annual {ACM-SIAM} Symposium on Discrete
  Algorithms ({SODA})}, pages 1038--1050. {SIAM}, 2017.

\bibitem{BJN}
N.~Buchbinder, K.~Jain, and J.~Naor.
\newblock Online primal-dual algorithms for maximizing ad-auctions revenue.
\newblock In {\em Proc.\ 15th Annual European Symposium on Algorithms {(ESA)}},
  Springer LNCS, Volume 4698, pages 253--264, 2007.

\bibitem{CDKL}
K.~Chaudhuri, C.~Daskalakis, R.D. Kleinberg, and H.~Lin.
\newblock Online bipartite perfect matching with augmentations.
\newblock In {\em Proc.\ 28th {IEEE} International Conference on Computer
  Communications (INFOCOM)}, pages 1044--1052, 2009.

\bibitem{CW}
I.R. Cohen and D.~Wajc.
\newblock Randomized online matching in regular graphs.
\newblock In {\em Proc.\ 29th Annual {ACM-SIAM} Symposium on Discrete
  Algorithms (SODA)}, pages 960--979. {SIAM}, 2018.

\bibitem{COS}
R.~Cole, K.~Ost, and S.~Schirra.
\newblock Edge-coloring bipartite multigraphs in {$O(E \log D)$} time.
\newblock {\em Comb.}, 21(1):5--12, 2001.

\bibitem{CL}
J.~Csima and L.~Lov{\'{a}}sz.
\newblock A matching algorithm for regular bipartite graphs.
\newblock {\em Discret. Appl. Math.}, 35(3):197--203, 1992.

\bibitem{DH}
N.R. Devanur and T.P. Hayes.
\newblock The adwords problem: online keyword matching with budgeted bidders
  under random permutations.
\newblock In {\em Proc.\ 10th {ACM} Conference on Electronic Commerce ({EC})},
  pages 71--78. {ACM}, 2009.

\bibitem{DSA}
N.R. Devanur, B.\ Sivan, and Y.\ Azar.
\newblock Asymptotically optimal algorithm for stochastic adwords.
\newblock In {\em Proc.\ 13th {ACM} Conference on Electronic Commerce {(EC)}},
  pages 388--404. {ACM}, 2012.

\bibitem{GKK}
A.~Goel, M.~Kapralov, and S.~Khanna.
\newblock Perfect matchings in {$O(n\log n)$} time in regular bipartite graphs.
\newblock {\em {SIAM} J. Comput.}, 42(3):1392--1404, 2013.

\bibitem{GKP}
R.L. Graham, D.E. Knuth, and O.~Patashnik.
\newblock {\em Concrete Mathematics: A Foundation for Computer Science}.
\newblock Addison-Wesley, 1989.

\bibitem{GKKV}
E.F. Grove, M.{-}Y. Kao, P.~Krishnan, and J.S. Vitter.
\newblock Online perfect matching and mobile computing.
\newblock In {\em Proc.\ 4th International Workshop on Algorithms and Data
  Structures (WADS)}, LNCS, Volume 955, pages 194--205. Springer, 1995.

\bibitem{HVV}
V.~Heikkala, M.K. Vamanamurthy, and M.~Vuorinen.
\newblock Generalized elliptic integrals.
\newblock {\em Comput. Methods Funct. Theory}, 9(1):75--109, 2009.

\bibitem{JL}
P.\ Jaillet and X.\ Lu.
\newblock Online stochastic matching: New algorithms with better bounds.
\newblock {\em Math. Oper. Res.}, 39(3):624--646, 2014.

\bibitem{KP}
B.~Kalyanasundaram and K.~Pruhs.
\newblock An optimal deterministic algorithm for online b-matching.
\newblock {\em Theor. Comput. Sci.}, 233(1-2):319--325, 2000.

\bibitem{KVV}
R.M. Karp, U.V. Vazirani, and V.V. Vazirani.
\newblock An optimal algorithm for on-line bipartite matching.
\newblock In {\em Proc.\ 22nd Annual {ACM} Symposium on Theory of Computing
  (STOC)}, pages 352--358, 1990.

\bibitem{MY}
M.~Mahdian and Q.~Yan.
\newblock Online bipartite matching with random arrivals: {An} approach based
  on strongly factor-revealing {LP}s.
\newblock In {\em Proc.\ 43rd {ACM} Symposium on Theory of Computing (STOC)},
  pages 597--606, 2011.

\bibitem{MOGS}
V.H.\ Manshadi, S.~Oveis Gharan, and A.\ Saberi.
\newblock Online stochastic matching: Online actions based on offline
  statistics.
\newblock {\em Math. Oper. Res.}, 37(4):559--573, 2012.

\bibitem{MSVV}
A.~Mehta, A.~Saberi, U.V. Vazirani, and V.V. Vazirani.
\newblock Adwords and generalized online matching.
\newblock {\em J. {ACM}}, 54(5):22, 2007.

\bibitem{NW}
J.~Naor and D.~Wajc.
\newblock Near-optimum online ad allocation for targeted advertising.
\newblock {\em {ACM} Trans. Economics and Comput.}, 6(3-4):16:1--16:20, 2018.

\bibitem{SCH}
A.~Schrijver.
\newblock Bipartite edge coloring in {$O(\Delta m)$} time.
\newblock {\em {SIAM} J. Comput.}, 28(3):841--846, 1998.

\bibitem{V}
V.V. Vazirani.
\newblock Randomized online algorithms for {Adwords}.
\newblock {\em CoRR}, abs/2107.10777, 2021.

\end{thebibliography}

\end{document}